\documentclass[11pt,leqno,letterpaper]{article}
\usepackage[utf8]{inputenc}
\usepackage[margin=1in]{geometry}
\usepackage{array}
\usepackage{multirow}
\usepackage{amsmath,amssymb,amsthm,amsfonts}
\usepackage[lined,ruled,,boxed, vlined, linesnumbered]{algorithm2e}

\usepackage{nicefrac}
\usepackage{dsfont}
\usepackage[colorlinks,citecolor=blue,bookmarks=true,breaklinks]{hyperref}
\usepackage{url}

\usepackage{breakurl}

\usepackage{tikz}
\usetikzlibrary{math} 
\usetikzlibrary{arrows.meta}

\usepackage{caption}
\usepackage{subcaption}
\usepackage{tikzsymbols}
\usepackage{tabularx} 
\usepackage[inline]{enumitem}
\theoremstyle{plain}
\newtheorem{theorem}{Theorem}[section]
\newtheorem{lemma}[theorem]{Lemma}
\newtheorem{corollary}[theorem]{Corollary}

\theoremstyle{definition}
\newtheorem{definition}[theorem]{Definition}

\newcommand{\ignore}[1]{}

\newcommand{\RS}[1]{ {#1} }
\newcommand{\CS}[1]{{\mathrm{#1}}}

\newcommand{\E}[1]{\text{\normalfont E}\left[ #1 \right]}
\makeatletter
\newcommand{\@giventhatnostar}[2]{#1\;\middle|\;#2}
\newcommand{\@giventhatstar}[3][]{#1(#2\;#1|\;#3#1)}
\newcommand{\given}{\@ifstar\@giventhatstar\@giventhatnostar}
\makeatother
\global\long\def\AOS{\text{AOS}}

\global\long\def\IID{\text{i.i.d.}}
\global\long\def\ALG{\text{\normalfont ALG}}
\global\long\def\OPT{\text{\normalfont OPT}}
\newcommand{\set}[1]{\left\{ #1 \right\}}

\newcommand{\Csample}{\CS{U}} 
\newcommand{\Rsample}{\RS{L}'} 

\newcommand{\CEsample}{\CS{F}} 
\newcommand{\REsample}{\RS{E}'} 

\newcommand{\MOnline}{M}
\newcommand{\Greedy}{\textsc{Greedy}}

\usepackage{multicol}

\newcommand{\cC}{{\cal C}}

\newcommand{\cE}{{\cal E}}

\newcommand{\cI}{{\cal I}}
\newcommand{\cJ}{{\cal J}}

\newcommand{\cP}{{\cal P}}

\newcommand{\cU}{{\cal U}}
\newcommand{\cV}{{\cal V}}

\makeatletter
\newenvironment{proofw}{\par
  \pushQED{\qed}%
  \normalfont \topsep7\p@\@plus6\p@\relax
  \trivlist
  \item[]\ignorespaces
}{%
  \popQED\endtrivlist\@endpefalse
}
\makeatother

\title{Online Weighted Matching with a Sample}

\author{
    Haim Kaplan%
    \thanks{Blavatnik School of Computer Science, Tel Aviv University, Israel.
    Email: \texttt{haimk@tau.ac.il}. 
    Supported by ISF grant.\ 1595-19 and the Blavatnik Family Foundation.}
    \and
    David Naori%
    \thanks{Computer Science Department, Technion, Israel.
    Emails: \texttt{\{dnaori,danny\}@cs.technion.ac.il}.} 
    \and Danny Raz%
    \footnotemark[2]
}
\date{}

\begin{document}

\begin{titlepage}
    \thispagestyle{empty}
    \maketitle
    \begin{abstract}
        \thispagestyle{empty}
        We study the greedy-based online algorithm for edge-weighted matching with (one-sided) vertex arrivals in bipartite graphs, and edge arrivals in general graphs. This algorithm was first studied more than a decade ago by Korula and P\'al for the bipartite case in the random-order model. While the weighted bipartite matching problem is solved in the random-order model, this is not the case in recent and exciting online models in which the online player is provided with a sample, and the arrival order is adversarial. The greedy-based algorithm is arguably the most natural and practical algorithm to be applied in these models. Despite its simplicity and appeal, and despite being studied in multiple works, the greedy-based algorithm was not fully understood in any of the studied online models, and its actual performance remained an open question for more than a decade.

We provide a thorough analysis of the greedy-based algorithm in several online models. For vertex arrivals in bipartite graphs, we characterize the exact competitive-ratio of this algorithm in the random-order model, for any arrival order of the vertices subsequent to the sampling phase (adversarial and random orders in particular). We use it to derive tight analysis in the recent adversarial-order model with a sample (AOS model) for any sample size, providing the first result in this model beyond the simple secretary problem. Then, we generalize and strengthen the black box method of converting results in the random-order model to single-sample prophet inequalities, and use it to derive the state-of-the-art single-sample prophet inequality for the problem. Finally, we use our new techniques to analyze the greedy-based algorithm for edge arrivals in general graphs and derive results in all the mentioned online models. In this case as well, we improve upon the state-of-the-art single-sample prophet inequality.
    \end{abstract}
\end{titlepage}

\section{Introduction}\label{sec:intro}

We study the online edge-weighted matching problem with vertex arrivals in bipartite graphs, and edge arrivals in general graphs.
Through the lens of the standard worst-case competitive-analysis paradigm,
any non-trivial performance guarantees cannot be achieved for these problems~\cite{DBLP:conf/soda/AggarwalGKM11}. Therefore, different online models with additional restricting assumptions are used.
These assumptions restrict the input sequences in various ways, and as a consequence, narrow the scope of the derived guarantees.

In recent years, sample-based online models have been introduced with the goal of minimizing restricting assumptions and achieving robust guarantees for realistic conditions. These models assume that the online algorithm has an access to a sample which can be extracted from historical data or by other means. This sample gives the algorithm limited information about future occurrences.
One such model is the recent adversarial-order model with a sample (AOS) by Kaplan et al.~\cite{kaplan2020competitive},
which allows the input and the arrival order to be fully adversarial, as in the standard worst-case competitive-analysis,
but provides a sample to the online player in advance.
Another related model is the single-sample prophet inequality by Azar et al.~\cite{soda/AzarKW14}.

Sample-based online algorithms have a typical structure. Given a sample, when an online element arrives and a decision has to be made,
they base their decision only on information from the sample and the current element. Naturally, using information gathered from previous online rounds is risky, as the adversary controls the arrival order and can reveal deceiving information. 
A common scheme of this form uses an offline algorithm as a black box: At every online round, run the offline algorithm on an input that consists of the sample and the arriving element. Then, treat the arriving element based on the decisions of the offline algorithm. A natural offline algorithm to apply this scheme with is the greedy algorithm.

For edge-weighted bipartite matching with (one-sided) vertex arrivals, this greedy-based online algorithm was first studied by Korula and P{\` a}l~\cite{korula2009algorithms} in the random-order model. 
In the random-order model, a sample is retrieved from a prefix of the online sequence (sampling phase).

The greedy-based algorithm has many advantages. For example, it admits a convenient and efficient form of a price-threshold policy which makes it useful for posted-price mechanisms. Despite its simplicity and appeal, the greedy-based matching algorithm was not fully understood in any relevant online model, and its actual performance remained an open problem for more than a decade. 

Azar et al.~\cite{soda/AzarKW14} observed that many algorithms which were studied in the random-order model are sample-based.
As such, these algorithms can be analyzed so that their competitive-ratio holds for any arrival order of the online sequence after the sampling phase is done (historically, some works unintentionally accomplished this). Azar et al.\ called these algorithms and analyses
order-oblivious and gave a reduction that 
obtains an algorithm for single-sample prophet inequality from an order-oblivious algorithm, while maintaining the competitive-ratio of the latter.
This demonstrates the power of order-oblivious algorithms in a wider scope than the random-order model.

In this paper, we study the greedy-based online algorithm for edge-weighted matching with vertex arrivals in bipartite graphs, and edge-arrivals in general graphs. We improve and attain tight bounds on the performance of the greedy-based algorithm in several online models. Moreover, we beat the state-of-the-art results for both problems in the single-sample prophet inequality setting. Our analysis sheds light on the performance of the classical offline greedy algorithm on random induced sub-graphs (on vertices or edges), which has potential implications beyond the scope of this paper. 

\subsection{Our Contribution}

For vertex arrivals in bipartite graphs, we provide a tight order-oblivious analysis of the greedy-based algorithm.
More concretely,
we show that the greedy-based algorithm is exactly $\left(3-2\sqrt{2}\right)\approx \frac{1}{5.83}$-competitive, improving upon both the random-order competitive-ratio of $1/8$ and the order-oblivious competitive-ratio of $\frac{1}{13.5}$ by Korula and P{\'a}l~\cite{korula2009algorithms}.
Interestingly, we show that this is the exact competitive-ratio of the algorithm for any (even for the best-case) arrival order of the vertices after the sampling phase. Our analysis  characterizes the performance of the offline greedy algorithm on random induced sub-graphs in which each online vertex is drawn independently with probability $p$. We show that it gives $p/(1+p)$-approximation for the maximum weight matching in the entire graph and that this ratio is tight.

We continue and  derive a tight analysis of the greedy-based online algorithm in the recent adversarial-order model with a sample, providing the first result in this model outside the context of the simple secretary problem. We adopt the mathematically convenient independent sampling approach by Correa et al.~\cite{correa2021secretary}, in which every potential online item (vertex in our case) is drawn to the sample independently with probability $p$ ($\AOS p$). We show that for $p \leq 1/2$, the competitive-ratio of the algorithm is $p(1-p)$, and for $p > 1/2$, we retain the competitive-ratio of $1/4$. We also show that our analysis is tight for any $p$. 

As for the performance of greedy relative to any online algorithm in the $\AOS p$ model, we know the following. An upper-bound on the competitive-ratio of any online algorithm for the problem of $\min\set{p,1/2}$ can be derived from the upper-bound on the secretary problem by Kaplan et al.~\cite{kaplan2020competitive}. This shows that for $p \leq 1/2$, the greedy-based algorithm is at most a factor of $(1-p)$ away from the best possible competitive-ratio, and at most a factor of $1/2$ away for $p > 1/2$.

Subsequently, we generalize and strengthen the reduction of Azar et al.~\cite{soda/AzarKW14} in two different ways. 
First, we generalize the reduction to apply to batches. This allows to account for dependencies within each batch. For example, in our case, all edges incident to the same online vertex constitute a batch. Second, we reduce a stronger online model than the single-sample prophet inequality,
which we call the (batched) two-faced model. In this model,  we replace the underlying weight distributions with adversarial weights. The adversary needs to choose two weight vectors for each batch. Then, one random weight vector for each batch is given as a sample. The remaining vectors, one per batch, are used online to challenge the online player in an adversarial order.
We apply our new reduction to obtain a competitive-ratio of $\frac{1}{5.83}$ for edge-weighted bipartite-matching with vertex arrivals in the (batched) two-faced model, which improves upon a recent result of $1/8$ by D{\"u}tting et al.~\cite{dutting2021prophet} in the (weaker) single-sample prophet inequality setting. 

We note that previous results for single-sample prophet inequalities actually hold in the stronger two-faced model (albeit, not with batches)~\cite{soda/AzarKW14, rubinstein2020optimal, caramanis2021single, dutting2021prophet}. In particular, Rubinstein et al.~\cite{rubinstein2020optimal} showed that for the classical single choice prophet inequality, the optimal competitive-ratio of $1/2$ with full knowledge of the distributions can be achieved in the two-faced model (and hence with a single sample). 
In contrast, in the case of edge-weighted bipartite matching with vertex arrivals, we show that there is a separation between the full knowledge prophet inequality and the two-faced model. We prove an upper-bound of $2/5$ on the competitive-ratio of any algorithm in the two-faced model, whereas Feldman et al.~\cite{feldman2014combinatorial} gave an optimal $1/2$-competitive algorithm for the problem with full knowledge of the distributions.

We then use our techniques to analyze the greedy-based algorithm for edge-weighted matching with edge arrivals in general graphs. As in the vertex arrivals case, we begin by providing an order-oblivious analysis of the algorithm. Interestingly, we show that the offline greedy algorithm on a random induced sub-graph to which each edge is drawn independently with probability $p$, is exactly $\min\set{p,1/2}$-approximation of the maximum matching in the entire graph. In particular, greedy retains its $1/2$ approximation-ratio even when each edge is discarded from the graph with probability $1/2$, which is the best possible approximation-ratio in this setting. 
Then, we use our reduction to obtain a competitive-ratio of $\frac{1}{11.66}$ in the two-faced model, improving upon recent results of $1/16$ by D{\"u}tting et al.~\cite{dutting2021prophet} for the single-sample prophet inequality setting. Finally, we derive results in the $\AOS p$ model for all values of $p$. 
\subsection{Related Work}

Online matching problems have been studied extensively over the last three decades. For unweighted bipartite graphs with one-sided vertex arrivals, the celebrated result of Karp et al.~\cite{karp1990optimal} shows that there is an optimal $(1-1/e)$-competitive online algorithm in the standard worst-case adversarial-order model.
The edge-weighted problem has been studied in various different online models, the most relevant of which are the random-order model, and stochastic online models such as the $\IID$ model from a known or unknown distribution, and prophet inequalities~\cite{mehta2013online, feldman2014combinatorial}. We survey the studies that are most related to our work.

In the random-order model, Kesselheim et al.~\cite{DBLP:conf/esa/KesselheimRTV13} gave an optimal $1/e$-competitive algorithm for vertex arrivals in bipartite graphs, reaching the best competitive-ratio possible, even for the special case of the classical secretary problem. Later, in pursuit of an algorithm that uses only ordinal information, Hofer and Kodric~\cite{hoefer2017combinatorial} studied a similar algorithm which is based on the offline greedy algorithm, but slightly different from the greedy-based algorithm we study in this paper. They showed that it achieves a competitive-ratio of $1/(2e)$. As opposed to the sample-based algorithms described in the introduction, both algorithms base their decisions on all the revealed information (and not only on the sample), which makes them susceptible to adversarial arrival orders.
For edge arrivals in general graphs, Kesselhiem et al.~showed a $1/(2e)$-competitive algorithm which was recently improved by Ezra et al.~\cite{ezra2020secretary} who gave a $1/4$-competitive (exponential time) online algorithm.
While the random order assumption allows achieving desirable constant competitive-ratios, 
in most realistic conditions, online vertices and edges arrive in a coordinated fashion over time, and 
the arrival order is far from being uniformly random.

The $\IID$ model from a known or unknown distribution relies on stronger assumptions. Indeed, the competitive-ratio of any algorithm may only improve in the $\IID$ model compared to the random-order model~\cite{mehta2013online}.
For vertex arrivals in bipartite graphs, it restricts each vertex to be drawn independently from a distribution over vertex types, which is either known or unknown to the online player. In the known distribution case, Haeupler et al.~\cite{haeupler2011online} presented a simple $(1-1/e)$-competitive algorithm for general distributions, and a better ratio for the special case of integral rates, which was later improved to $0.705$ by Brubach et al.~\cite{brubach2016new}. 

In the (batched) prophet inequality setting, for vertex arrivals in bipartite graphs, 
all edges incident to the same (online) vertex constitute a batch. The edge weights of each batch are drawn independently from a particular joint distribution for the batch.
All distributions are known to the online player in advance, and the vertices arrive online in an adversarial order. Feldman et al.~\cite{feldman2014combinatorial} gave an optimal $1/2$-competitive algorithm in this model with an oblivious adversary who chooses the arrival order without knowing the weight realizations. As in the random-order model, this result matches the best competitive-ratio for the special case of the single choice prophet inequality. For edge arrivals in general graphs, Ezra et al.~\cite{ezra2020online} presented a $0.337$-competitive algorithm (also for an oblivious adversary). While the prophet inequality provides robust guarantees against (oblivious) adversarial arrival order, it relies on 
the rather strong assumption that the distributions are known exactly.

In recent years, there is a growing interest in the study of online models with limited prior data. Many of these studies are focused on the basic single-choice prophet inequality or the secretary problem. 
Correa et al.~\cite{DBLP:conf/ec/CorreaDFS19} studied the single-choice prophet inequality in an $\IID$ model from an unknown distribution, in which the online player is equipped with a limited number of training samples from the distribution. They presented a $(1-1/e)$-competitive algorithm that uses $n$ training samples, which was later improved by \cite{correa2020two} and \cite{correa2020streaming}. 

As mentioned before, Kaplan et al.~\cite{kaplan2020competitive} introduced the adversarial-order model with a sample ($\AOS$ model), in which the input is fully adversarial, but a uniformly random sample of limited size is revealed to the online player upfront for the purpose of learning. Then, the remaining part of the input arrives online in an adversarial order.
Later, Correa et al.~\cite{correa2021secretary} introduced the $\AOS$ model with $p$-sampling ($\AOS p$), which differs from the $\AOS$ model only in the sampling procedure: Each online element is drawn to the sample independently with probability $p$. This sampling procedure makes the model more mathematically convenient. Both the $\AOS$ and $\AOS p$ were used to study the simple secretary problem.

Azar et al.~\cite{soda/AzarKW14} studied prophet inequalities with limited information, in which the assumption about the knowledge of the distributions is replaced with a limited number of samples from each distribution. They showed a $\frac{1}{6.75}$-competitive algorithm for matching in bipartite graphs with constant maximum-degree $d$, that uses $d^2$ samples from each distribution. As discussed before, Azar et al.~\cite{soda/AzarKW14} designed a method of converting an order-oblivious algorithm to an algorithm in the single-sample prophet inequality setting. The reduction does not use all the available information in the input. Therefore, there is a hope to attain better bounds by analyzing algorithms directly in the single-sample prophet inequality setting. Indeed, Rubinstein et al.~\cite{rubinstein2020optimal} obtained an optimal $1/2$-competitive algorithm for the single-choice prophet inequality with a single sample by a direct analysis, compared to a competitive-ratio of $1/4$ that can be obtained through the reduction with a simple order-oblivious algorithm. Moreover, we show that $1/4$ is in fact the best possible guarantee one can get through the reduction.

Very recently, Caramanis et al.~\cite{caramanis2021single} and D{\"u}tting et al.~\cite{dutting2021prophet} studied the greedy-based algorithm for the edge-weighted matching problem with vertex arrivals in bipartite graphs, and edge arrivals in general graphs in the single-sample prophet inequality setting with independent edge values.
Inspired by the work of Rubinstein et al.~\cite{rubinstein2020optimal}, both works aimed to improve the known bounds by analyzing the greedy-based algorithm directly in the single-sample prophet inequality setting and avoid going through the reduction by Azar et al.~\cite{soda/AzarKW14}. While we do not analyze algorithms directly in the single-sample prophet inequality setting, and go through a similar reduction, we still improve upon their results and achieve the state-of-the-art competitive-ratio for both vertex and edge arrivals.

Another related model with limited prior-knowledge was studied by Kumar et al.~\cite{kumar2018semi} for the unweighted bipartite matching. They considered a semi-online model in which a sub-graph is known in advance (the known sub-graph also arrives online) and developed algorithms with competitive-ratio that improves as the size of the known fraction of the graph increases.

In a bit more distant settings, several exciting advances were recently made. Farbach et al.~\cite{Fahrbach-focs2020} studied the edge-weighted bipartite-matching problem with vertex arrivals in the free-disposal model, where each offline vertex can be matched any number of times, but only the heaviest edge counts. In this model, greedy achieves a competitive-ratio of $1/2$. They show how to break the $1/2$ barrier by introducing an online correlated selection technique. Additional advances were made for general arrivals in the unweighted and vertex-weighted cases (see~\cite{huang2018match,huang2019tight,huang2019online,gamlath2019online} for example). We also refer the interested reader to the extensive survey by Aranyak Mehta~\cite{mehta2013online}.
\subsection{Organization of the Paper}
In Section~\ref{sec:vertex_bipartite} we present our tight
order-oblivious analysis of the greedy-based algorithm for vertex arrivals in bipartite graphs.
In Section~\ref{sec:bipartite_aos_p} we continue with vertex arrivals in bipartite graphs in the $\AOS p$ model. We begin with a formal definition of the problem, then we present tight bounds on the performance of the greedy-based algorithm, for any value of $p$, that are derived from our order-oblivious analysis.
In Section~\ref{sec:batched} we establish formal and general definitions for batched online selection problems, and present our improved black box reduction from order-oblivious algorithms to two-faced algorithms. Then, we use it to derive results for the two-faced bipartite matching with vertex arrivals, and finish by proving an upper-bound on the competitive-ratio of any algorithm for the problem.
Finally, in Section~\ref{sec:edge_general} we consider edge arrivals in general graphs. We begin with an order-oblivious analysis of the greedy-based algorithm, and then use it to derive results in the two-faced model and the $\AOS p$ model.
\section{Vertex Arrivals in Bipartite Graphs}\label{sec:vertex_bipartite}

In this section, we present a tight order-oblivious analysis of the greedy-based algorithm in the random-order model. Recall that an analysis in the random-order model is called order-oblivious, if it holds for any arrival order of the online sequence after the sampling phase is done. We provide a more formal definition of order-oblivious algorithms in Section~\ref{sec:batched} after establishing general definitions of online selection problems. The order-oblivious analysis in this section is also used in Section~\ref{sec:bipartite_aos_p} to derive tight results in the adversarial order model with a sample ($\AOS p$). 

We begin with a formal definition of the problem in the random-order model.
In the random-order online (weighted) bipartite matching problem with (one-sided) vertex arrivals, an adversary chooses a bipartite graph $\CS{G} = (\CS{L}, \CS{R}, \CS{E})$ with non-negative edge weights $w : \CS{E} \rightarrow \mathbb{R}_{\geq 0}$.\footnote{We assume that a consistent tie-breaker is available so that the set of weights $\set{w(e)}_{e\in\CS{E}}$ is totally ordered. Throughout the paper, when edge weights are compared, it is implicitly assumed that this tie-breaker is applied. }
The right side vertices $\CS{R}$ and the cardinality of the left side $\CS{|L|}$ are revealed to the online algorithm upfront. Then, the vertices of $\CS{L}$ arrive one-by-one in a uniformly random order. When a vertex $u \in \CS{L}$ arrives, its incident edges and their weights are revealed. At this point, the online algorithm must either match $u$ to an available (unmatched) neighbor in $\CS{R}$, or leave $u$ unmatched. The decision is permanent, and must be made before the next online vertex arrives. The goal is to maximize the expected total weight of the produced matching.

For an algorithm $\ALG$ and an input instance $\cI$,
let $\ALG(\cI)$ be the random variable that gets the total weight of the matching computed by the algorithm, and let $\OPT(\cI)$ be the weight of a maximum matching in $\CS{G}$. With a slight abuse of notation, we also use $\ALG(\cI)$ and $\OPT(\cI)$ to refer to the respective matchings (as sets of edges, and not only to their weights). When $\mathcal{I}$ is clear from the context, we omit it from the notation and write, for example, $\ALG$ instead of $\ALG(\mathcal{I})$.

An algorithm $\ALG$ is called $c$-competitive, if for any input instance $\cI$, $\E{\ALG(\cI)} \geq c \cdot \OPT(\cI)$, where the expectation is taken over the random arrival order of the vertices, and the internal randomness of the algorithm.

We denote by $\textsc{Greedy}$ the well-known offline greedy algorithm for weighted matching: 
$\Greedy$ processes the edges of the graph one-by-one in a non-increasing order of their weight. When it processes $e = (u,v)$, it adds $e$ to its matching if no edge incident to $u$ or $v$ was added before.

In this section we provide a tight order-oblivious analysis for the fundamental greedy-based online algorithm,
which was first studied by Korula and P{\`a}l~\cite{korula2009algorithms}. The algorithm works as follows: It begins with a sampling phase of random length $k = Binom(|\CS{L}|,p)$ in which it only collects the incoming vertices into a sample $\RS{L}'$. Afterwards, when a new vertex $u$ arrives, the algorithm selects (at most) one of its incident edges as a candidate for its output matching. It determines the candidate edge by computing the offline greedy matching on the sub-graph induced by the vertices of the offline side $\CS{R}$, the sample set $\RS{L'}$, and the arriving vertex $u$. I.e. $\textsc{Greedy}(\CS{G}[L' \cup \set{u} \cup \CS{R}])$. The edge incident to $u$ in this matching, $(u,r)$, is the algorithm's selected candidate (if there is no such edge, no candidate is chosen). Then, the candidate edge is added to the output matching if no previous online vertex was matched to $r$.  For a formal description see Algorithm~\ref{alg:greedy_v_a}.\footnote{Equivalently, Algorithm~\ref{alg:greedy_v_a} can be described with price thresholds on the vertices of $\CS{R}$ as in~\cite{korula2009algorithms}, which provides a more practical way of implementing it.  The price on each vertex $r \in \CS{R}$ is simply the weight of the edge incident to $r$ in the greedy matching on the sample $\Greedy(\CS{G}[L' \cup \CS{R}])$. Then, when an online vertex $u$ arrives, its candidate edge is the heaviest edge incident to $u$ whose weight exceeds the price of its right side vertex.}

In this section, we refer only to induced sub-graphs of $\CS{G}$ that contain the entire right side $\CS{R}$. Hence, to abbreviate the notation, for $\CS{U} \subseteq \CS{L}$ we write $\CS{G}[\CS{U}]$ instead of $\CS{G}[\CS{U} \cup \CS{R}]$.

\IncMargin{1em}
\begin{algorithm}
\caption{Order-Oblivious Greedy-Based Vertex Arrivals in Bipartite Graphs}
\label{alg:greedy_v_a}
$k \leftarrow Binom(|\CS{L}|,p)$\;

Let $\RS{L}' \subseteq \CS{L}$ be the first $k$ vertices that arrive online\tcp*{sampling phase}

$ \MOnline \leftarrow \emptyset$\;

\For {a vertex $u_{\ell}$ that arrives at round $\ell > k$ } {
$\RS{G}_\ell \leftarrow \CS{G}[{\RS{L'} \cup \set{u_\ell}}]$\;
    $\RS{M}_\ell \leftarrow \textsc{Greedy}(\RS{G}_\ell)$\;
    \uIf{$u_\ell$ is matched in $\RS{M}_\ell$} {
        Let $(u_\ell, r_\ell) \in \RS{M}_\ell$ be the corresponding edge\tcp*{a candidate edge}\label{line:candidate_edge}
        \If {$r_\ell$ is not matched in $\MOnline$} {
	        $ \MOnline \leftarrow \MOnline \cup \set{(u_\ell, r_\ell)}$\;
        }
    }
}
\Return{\RS{M}}
\end{algorithm}
\DecMargin{1em}

For the analysis, given an input graph $\CS{G}$, we define an auxiliary directed graph, which will also be useful for the case of edge arrivals in general graphs.
\begin{definition}[directed line-graph]\label{def:directed_line_graph}
Given an edge-weighted graph $\CS{G} = (\CS{V},\CS{E})$, $w : \CS{E} \rightarrow \mathbb{R}_{\geq 0}$, the \textit{directed line-graph} of $\CS{G}$ is a vertex-weighted directed graph $\CS{G_D} = (\CS{V_D},\CS{E_D})$. To avoid confusion between $\CS{G_D}$ and $\CS{G}$, we refer the elements of $\CS{V_D}$ and $\CS{E_D}$ by nodes and arcs, respectively. Each node of $\CS{V_D}$ is associated with an edge of $\CS{E}$; $\CS{V_D} = \set{v_e : e \in \CS{E}}$. The weight of each node $v_e \in \CS{V_D}$ is $w(e)$ and there is an arc $v_e \rightarrow v_{e'}$ if $e$ and $e'$ share a common vertex, i.e., $e \cap e' \neq \emptyset$ and $w(e') < w(e)$.

For $u \in \CS{V}$, $\CS{V_D}(u) = \set{v_e \in \CS{V_D} : u \in e}$ is called the \textit{cluster} of $u$. In words, $\CS{V_D}(u)$ is the set of nodes that correspond to edges incident to $u$.
\end{definition}

Clearly, $\CS{G_D}$ is a directed acyclic graph. Furthermore,
let $v_{e_1}, \dots, v_{e_m}$ be the nodes of $\CS{G_D}$ ordered in a non-increasing order of weight. Then, $v_{e_1}, \dots, v_{e_m}$ is a topological ordering of $\CS{G_D}$. 

For a given bipartite-graph $\CS{G} = (\CS{L},\CS{R},\CS{E})$, we consider its directed line-graph $\CS{G_D} = (\CS{V_D},\CS{E_D})$. We focus our attention on the vertices of the left side $\CS{L}$. Observe that $\set{\CS{V_D}(u)}_{u \in \CS{L}}$ is a partition of $\CS{V_D}$. Thus, we can identify each node $v_{(u,r)} \in \CS{V_D}$, with the cluster $\CS{V_D}(u)$ of its left side vertex $u$. 

We assume that the sampling $\RS{L'} \subseteq \CS{L}$ is done gradually by coloring the clusters $\set{\CS{V_D}(u)}_{u \in \CS{L}}$. Each cluster is colored red independently with probability $p$, and blue otherwise. When a cluster is colored, all its nodes are colored with the same color. Then, $\RS{L'}$ is the set of vertices $u$ whose cluster $\CS{V_D}(u)$ is red. 

We define the random coloring process of the clusters and the notion of an \textit{active} node
inductively over $v_{e_1},\dots,v_{e_n}$: $v_{e_1}$ is always active, and its cluster is colored randomly (red with probability $p$ and blue otherwise).
Given the active/inactive status of the nodes in $v_{e_1},\dots,v_{e_{i-1}}$, and the colors of the active nodes among them, $v_{e_{i}}$ is active if it is uncolored, and there is no incoming arc to $v_{e_{i}}$ from an active red node.
If $v_{e_{i}}$ is active, its cluster is colored randomly. At the end of the process, each remaining uncolored cluster is also colored randomly. Note that the cluster is colored when we reach the first node which is active in the cluster. At the end we color clusters without any active node.

Observe that indeed, through this coloring process an independent random decision is made to color each cluster red with probability $p$ and blue otherwise.

\begin{lemma}\label{lem:red}
Let $e_i = (u,r)$ for $u \in \CS{L}$, then $e_i \in \Greedy({\CS{G}[\Rsample \cup \{u\}]})$ if and only if $v_{e_i}$ is active.
\end{lemma}

\begin{proof}
We prove this by induction on $i$. For $e_1 = (u,r)$, clearly $e_1 \in \Greedy({\CS{G}[\Rsample \cup \{u\}]})$ and $v_{e_1}$ is active by definition. Now $e_i = (u,r)$ is added by $\Greedy({\CS{G}[\Rsample \cup \{u\}]})$ if and only if no heavier edge (in $e_1,\dots,e_{i-1}$), incident to $u$ or $r$, is added by $\Greedy({\CS{G}[\Rsample \cup \{u\}]})$. By the induction hypothesis, this happens if and only if there is no incoming arc to $v_{e_i}$ from an active node in $\{\CS{V_D}(w)\}_{w \in \Rsample \cup \{u\}}$. This happens if and only if $v_{e_i}$ has no incoming arcs from active red nodes (which are the nodes in $\{\CS{V_D}(w)\}_{w \in \Rsample}$), and when the coloring process reaches $v_{e_i}$, it is uncolored (i.e., nodes in $\CS{V_D}(u)$ that precede $v_{e_i}$ are not active). That is, by definition, $v_{e_i}$ is active.
\end{proof}

Next, we bound the expected performance of the algorithm in terms of the expected performance of $\Greedy$ on the random sample $\Rsample$. We first express the expected performance of $\Greedy$ using the probabilities of the nodes $v_{e_1},\dots,v_{e_m}$ to be active.

\begin{lemma}\label{lem:greedy_v_a_prob_active}
$\E{\Greedy(\CS{G}[\RS{L}'])} =  p \sum_{i=1}^{m} w(e_i) \Pr[v_{e_i}\text{ is active}]$
\end{lemma}

\begin{proof}
For $e_i = (u,r)$, we have $u \in \Rsample$ if and only if $v_{e_i}$ is red. Together with Lemma~\ref{lem:red}, we have $\Pr[e_i \in \Greedy({\CS{G}[\Rsample]})] = \Pr[v_{e_i}\text{ is active and red}]$. By the definition of the coloring process, when $v_{e_i}$ is active, it is colored red independently with probability $p$. Hence, we have
\begin{align*}
    \Pr[e_i \in \Greedy({\CS{G}[\Rsample]})] =  \Pr\left[\given{v_{e_i} \text{ is red}}{  v_{e_i}\text{ is active}}\right] \Pr[v_{e_i}\text{ is active}] = p \cdot \Pr[v_{e_i}\text{ is active}], 
\end{align*}
and so
\begin{align}
    \E{\Greedy(\CS{G}[\RS{L}'])} &= \sum_{i=1}^{m} w(e_i) \Pr[e_i \in \Greedy({\CS{G}[\Rsample]})] =p \sum_{i=1}^{m} w(e_i) \Pr[v_{e_i}\text{ is active}].\qedhere
\end{align}
\end{proof}

We now proceed to analyze the performance of the algorithm. We first point out a useful observation about the adversary. The sample $\Rsample$ determines a candidate edge (see line~\ref{line:candidate_edge} of Algorithm~\ref{alg:greedy_v_a}) for each $u \in \CS{L} \setminus \Rsample$ (possibly none). For each $r \in \CS{R}$ the algorithm matches $r$ to the first vertex $u \in \CS{L} \setminus \RS{L'}$ with a candidate edge $(u,r)$ that arrives (after the sampling phase). Therefore, in the worst-case arrival order of the vertices in $\CS{L} \setminus \RS{L}'$, among all vertices with a candidate edge incident to $r$ in $\CS{L} \setminus \RS{L'}$, the vertex with the lightest edge to $r$ arrives first. We, therefore, assume without loss of generality that the adversary reveals the vertices in this order. In other words, for $r \in \CS{R}$, a candidate edge $(u,r)$ is added by the algorithm if and only if $(u,r)$ is the lightest candidate edge incident to $r$.

Next, we lower bound the probability that $e_i = (u,r)$ is added by the algorithm. We begin by observing that $v_{e_i}$ is active and blue if and only if $e_i$ is a candidate edge: For $e_i$ to be a candidate edge we need $u \in \CS{L} \setminus \Rsample$ and $e_i \in \Greedy({\CS{G}}[\Rsample \cup \set{u}] )$. First we have $u \in \CS{L} \setminus \Rsample$ if and only if $v_{e_i}$ is blue, and by Lemma~\ref{lem:red}, $e_i \in \Greedy({\CS{G}}[\Rsample \cup \set{u}] )$ if and only if $v_{e_i}$ is active.

To account for the contribution of $e_i$ to $\MOnline$ (the output matching of the algorithm) we define the notion of a qualifying edge. We say that $e_i$ \textit{qualifies} if $v_{e_i}$ is active, blue, and it has no outgoing arcs to other active blue nodes.
By our observations above, when $e_i$ qualifies, it is a candidate edge and there are no candidate edges of smaller weight that intersect $e_i$ (a candidate edge $e_j$ of smaller weight that intersect $e_i$ corresponds to an active and blue node $v_{e_j}$ with an arc $v_{e_i} \rightarrow v_{e_j}$). Hence by our assumption on the adversary, when $e_i$ qualifies, it is added by the algorithm to $\MOnline$. We have
\begin{align}\label{eq:alg_actives}
    \E{\ALG} &\geq \sum_{i=1}^{m} w(e_i) \Pr[e_i\text{ qualifies}].
\end{align}

We are now ready to relate the expected performance of the algorithm to the expected performance of $\Greedy$ on the sample.
\begin{lemma}\label{lem:alg_vs_greedy_v_a}
For $p \in [0,1]$, $\E{\ALG} \geq  (1-p)\E{\Greedy(\CS{G}[\RS{L}'])} $ for any arrival order of the vertices in $\CS{L} \setminus \Rsample$.
\end{lemma}

\begin{proof}
We lower bound the probability that an edge $e_i = (u,r)$ qualifies. Consider the coloring process until reaching $v_{e_i}$ such that $v_{e_i}$ is active and blue. For $e_i$ to qualify, $v_{e_i}$ must not have outgoing arcs to active blue nodes. The outgoing arcs of $v_{e_i}$ are only to later nodes in the process $v_{e_{i+1}},\dots,v_{e_m}$. Moreover, all the outgoing arcs of $v_{e_i}$ are to nodes in $\CS{V_D}(u) \cup \CS{V_D}(r)$.
Since $v_{e_i}$ is colored blue, all the nodes in $\CS{V_D}(u)$ are colored blue with it. So, nodes from $\CS{V_D}(u)$ later in the process cannot be active (as a node must be uncolored to be active).

We continue the coloring process until reaching an active node $v_{e_j}$ in $\CS{V_D}(r)$. If there is no such node, $v_{e_i}$ qualifies (with probability $1$). Otherwise, if $v_{e_j}$ is colored red, which happens independently with probability $p$, all future nodes in $\CS{V_D}(r)$ will be inactive (as $v_{e_j}$ has an outgoing arc to each one of them), and there will be no outgoing arcs from $v_{e_i}$ to active blue nodes in $\CS{V_D}(r)$.
We get that conditioned on the event that $v_{e_i}$ is active and blue, $e_i$ qualifies with probability at least $p$. To conclude,
\begin{align}
\begin{split}
    \Pr[e_i\text{ qualifies}] &=  \Pr[v_{e_i}\text{ is active and blue}] \cdot \Pr[{e_i} \text{ qualifies}\mid v_{e_i} \text{ is active and blue}] \\
    &\geq \Pr[v_{e_i}\text{ is active}] \cdot (1-p) p.
\end{split}
\label{eq:qualifies_prob_v_a}
\end{align}
Finally, by Lemma~\eqref{lem:greedy_v_a_prob_active} and Inequalities~\eqref{eq:alg_actives} and~\eqref{eq:qualifies_prob_v_a}, we obtain
\begin{align*}
    \frac{\E{\ALG}}{\E{\Greedy(\CS{G}[\Rsample])}} &\geq
    \frac{\sum_{e \in \CS{E}} w(e)\Pr[v_{e}\text{ is active}] \cdot (1-p) p}{\sum_{e \in \CS{E}} w(e)  \Pr[v_{e}\text{ is active}] \cdot p} = 1-p.\qedhere
\end{align*}
\end{proof}

Next, we establish a tight bound on the expected approximation-ratio of $\Greedy$ on a random sample of vertices.
\begin{lemma}\label{lem:p/(1+p)}
Let $\CS{G} = (\CS{L},\CS{R},\CS{E})$ be a bipartite graph and let $p \in [0,1]$. Let $\RS{L'} \subseteq \CS{L}$ such that each $u \in \CS{L}$ is drawn to $\RS{L'}$ independently with probability $p$. Then $\E{\Greedy(\CS{G}[\RS{L}'])} \geq \OPT \cdot  p / (1+p) $. Moreover, there is a sequence of bipartite graphs $\CS{G}_1,\CS{G}_2, \dots$ on which $\E{\Greedy(\CS{G}_k[\RS{L}'])}/ \OPT$ approaches $p/(1+p)$ as $k \rightarrow{\infty}$.
\end{lemma}

\begin{proof}
For a matching $M$ we split the weight of each edge $e = (u,r) \in M$ between its two endpoints: Let $c(M,u) = \alpha w(e)$ and $c(M,r) = (1-\alpha) w(e)$ for $\alpha \in [0,1]$. For an unmatched vertex $v$, we define $c(M,v) = 0$. We have $w(M) = \sum_{u \in \CS{V}} c(M,u)$. For convenience of notation, let $M_{L'} = \Greedy(\CS{G}[L'])$.

Let $M^*$ be a maximum matching in $\CS{G}$. Fix $e =(u,r) \in M^*$. We consider three possible events. 
\begin{enumerate*}[label=(\Roman*)]
    \item $e$ is in the greedy matching $M_{L'}$\label{event:a_1}
    \item There is an edge $e_u$ incident to $u$ and heavier than $e$ (i.e., $w(e_u) > w(e)$) in $M_{L'}$. And\label{event:a_u} 
    \item There is an edge $e_r$ incident to $r$ and heavier than $e$ in $M_{L'}$.\label{event:a_r}
\end{enumerate*} 
Let $A_{e}$ denote Event~\ref{event:a_1}. We consider three combinations of events~\ref{event:a_u} and~\ref{event:a_r}: Let $A_{e,u,r} =  \text{\ref{event:a_u}} \land \text{\ref{event:a_r}}$, i.e., $A_{e,u,r}$ is the event that both~\ref{event:a_u} and~\ref{event:a_r} occur. Let $A_{e,u} = \text{\ref{event:a_u}} \land \neg \text{\ref{event:a_r}}$, i.e., $A_{e,u}$ is the event that~\ref{event:a_u} occurs, and~\ref{event:a_r} does not, and finally, let $A_{e,r} =  \neg \text{\ref{event:a_u}} \land \text{\ref{event:a_r}}$. Observe that $A_{e}$, $A_{e,u,r}$, $A_{e,u}$ and $A_{e,r}$ are disjoint.

If $A_{e}$ occurs, $c(M_{L'},u) + c(M_{L'},r) = w(e)$, if $A_{e,u}$ occurs, $c(M_{L'},u) \geq \alpha w(e)$, if $A_{e,r}$ occurs, $c(M_{L'},r) \geq (1-\alpha) w(e)$, and if $A_{e,u,r}$ occurs, $c(M_{L'},u) + c(M_{L'},r) \geq w(e)$. We have 
\begin{align}
\begin{split}
    \E{\given{c(M_{L'},u) + c(M_{L'},r)}{ u \in L'}} &\geq 
     w(e) \Pr\left[\given{A_{e}}{ u \in \RS{L}'}\right] \\
     & \quad + w(e) \Pr\left[\given{A_{e,u,r}}{ u \in \RS{L}'}\right] \\
     & \quad + \alpha w(e) \Pr\left[\given{A_{e,u} }{ u \in \RS{L}'}\right] \\ & \quad + (1-\alpha) w(e) \Pr\left[\given{A_{e,r}}{ u \in \RS{L}'}\right].\label{eq:exp_u_in}
\end{split}
\end{align}
Observe that by the definition of greedy, given a sample $\RS{L}' = \CS{U}$ with $u\in \CS{U}$ such that $A_{e,r}$ occurs, $A_{e,r}$ still occurs when removing $u$ from $\CS{U}$.
So, the mapping $f(\CS{U})=\CS{U}\setminus \{u\}$ from the subspace $X=\{L'=\CS{U}\mid u\in \CS{U}\}$ to the subspace $Y=\{L'=\CS{U}\mid u\not\in \CS{U}\}$ is such that 1) if $A_{e,r}$ holds for $\CS{U}$ then it holds for $f(\CS{U})$ and 2) the probability of $\CS{U}$ in $X$ is the same as the probability of $f(\CS{U})$ in $Y$. Therefore, $\Pr\left[ \given{ A_{e,r} }{ u \notin \RS{L}' }\right] \geq \Pr\left[\given{A_{e,r}}{ u \in \RS{L}'}\right]$. We have
\begin{align}
\begin{split}
    \E{\given{c(M_{L'},u) + c(M_{L'},r)}{u \notin L'} } &\geq (1-\alpha) w(e) \Pr\left[ \given{A_{e,r}}{ u \notin \RS{L}'} \right] \\
    &\geq (1-\alpha) w(e) \Pr \left[ \given{A_{e,r}}{ u \in \RS{L}'} \right].\label{eq:exp_u_out}
\end{split}
\end{align}
By law of total expectation,
\begin{align}
\begin{split}
    \E{c(M_{L'},u) + c(M_{L'},r)} &=  p \E{c(M_{L'},u) + c(M_{L'},r) \mid u \in L' } \\
    &\quad +  (1-p)\E{c(M_{L'},u) + c(M_{L'},r) \mid u \notin L' },
\end{split}\label{eq:exp_total}
\end{align}
and by replacing Inequalities~\eqref{eq:exp_u_in} and~\eqref{eq:exp_u_out} in~\eqref{eq:exp_total} we obtain 
\begin{align}
\begin{split}
    \E{c(M_{L'},u) + c(M_{L'},r)} &\geq  p w(e) \left(\Pr\left[\given{A_{e}}{ u \in \RS{L}'}\right] + \Pr\left[\given{A_{e,u,r}}{ u \in \RS{L}'}\right] \right) \\
    &\quad + p \alpha w(e) \Pr\left[\given{A_{e,u}}{ u \in \RS{L}'}\right] + (1-\alpha) w(e) \Pr\left[\given{A_{e,r}}{ u \in \RS{L}'}\right] \\
    &\geq w(e)\min\set{p,p\alpha,(1-\alpha)} \cdot \left(     \Pr\left[\given{A_{e}}{ u \in \RS{L}'}\right]  + \Pr\left[\given{A_{e,u,r}}{ u \in \RS{L}'}\right] \right. \\
    & \quad \left. + \Pr\left[\given{A_{e,u}}{ u \in \RS{L}'}\right] + \Pr\left[\given{A_{e,r}}{ u \in \RS{L}'}\right] \right)\label{eq:exp_before_1}
\end{split}
\end{align}
Now observe that conditioned on $u \in L'$, exactly one of the events $A_e$, $A_{e,u,r}$, $A_{e,u}$ or $A_{e,r}$ must occur. Hence
\begin{align}
    \Pr\left[\given{A_{e}}{ u \in \RS{L}'}\right] + \Pr\left[\given{A_{e,u,r}}{ u \in \RS{L}'}\right] +  \Pr\left[\given{A_{e,u}}{ u \in \RS{L}'}\right] + \Pr\left[\given{A_{e,r}}{ u \in \RS{L}'}\right]  = 1.\label{eq:sum_to_one}
\end{align}
Using Inequalities~\eqref{eq:exp_before_1} and~\eqref{eq:sum_to_one}, we get that 
\begin{align*}
\begin{split}
    \E{c(M_{L'},u) + c(M_{L'},r)} &\geq \min\set{p, p\alpha , (1-\alpha)} w(e) \geq \frac{p}{1+p} w(e),
\end{split}
\end{align*}
where the last inequality is obtained for $\alpha = 1/ (1+p)$.
To sum up,
\begin{align*}
    \E{w(M_{L'})} &\geq \sum_{(u,r) \in M^*} \E{c(M_{L'}, u) + c(M_{L'}, r)}
    \geq \sum_{(u,r) \in M^*}  \frac{p}{1+p} w(u,r) 
    \geq \frac{p}{1+p}  w(M^*).
\end{align*}

We now proceed to show that the bound $p/(1+p)$ is tight. To this end, we construct a weighted version of the worst-case example for the classical RANKING algorithm for online unweighted bipartite matching by Karp et al.~\cite{karp1990optimal}. The weights are used only for tie-breaking against greedy. $\CS{G}_k = (\CS{L}, \CS{R}, \CS{E})$, $\CS{L} = \set{u_1,\dots,u_k}$ and $\CS{R} = \set{r_1,\dots, r_k}$. Each $u_i \in \CS{U}$ is connected only to $\set{r_1,\dots,r_{k-i+1}}$. All the edges incident to $u_i$ are heavier than the edges incident to $u_{i+1}$. Also, the edges of $u_i$ are decreasing in weight $w(u_i,r_1) > \dots > w(u_i,r_{k-i+1})$. Nevertheless, the weights of all edges are arbitrarily close to $1$. 

Clearly, there is a perfect matching in $\CS{G}_k$, and $\OPT$ is arbitrarily close to $k$. Now, consider the vertices $\CS{L}_t = \set{u_1, \dots ,u_{k - t}}$ and $\CS{R}_t = \set{r_1,\dots, r_t}$ for $t = \lceil k p/(1+p)  \rceil  $.
Greedy processes the edges incident to the vertices in $\RS{L}' \cap L_t$ first, and matches them to the first $|\RS{L}' \cap \CS{L}_t|$ vertices in $\CS{R}$. The vertices in $\set{u_{k-(t+1)},\dots,u_k}$ can only be matched to the remaining, unmatched vertices in $\CS{R}_t$. So greedy gets at most $\max\set{t, |L'\cap L_t|}$ edges in its matching. Thus, the ratio between $\E{\Greedy(\CS{G}_k[\RS{L}'])}$ and $\OPT$ is upper bounded by $\E{\max\set{t, |L'\cap L_t|}}/ k$. 

The expected number of vertices in $\RS{L}' \cap L_t$ is $\mu = (k - t)p \leq kp(1 - p/(1+p)) = kp/(1+p)$. By applying a Chernoff bound we get that $\Pr[|L'\cap L_t| > (1+\delta) kp/(1-p)] \leq e^{-\delta^2 \mu / 3} $. For $\delta = \sqrt{\log{k}/k}$ we obtain that for $k \rightarrow \infty$, with probability $1-o(1)$, $|L'\cap L_t| = (1+o(1))kp/(1+p)$. Thus, $\E{\max\set{t, |L'\cap L_t|}}/k\leq p/(1+p) + o(1)$.
\end{proof}
We note that the lower-bound in Lemma~\ref{lem:p/(1+p)} can also be proved using the technique of randomized primal-dual analysis by Devanur et al.~\cite{devanur2013randomized}.

We now derive the competitive-ratio of the algorithm.
\begin{theorem}\label{thm:greedy_v_a_c_r}
For $p \in [0,1]$, Algorithm~\ref{alg:greedy_v_a} is order-oblivious $p(1-p)/(1+p)$-competitive.
\end{theorem}

\begin{proof}
Let $p \in [0,1]$. By combining Lemma~\ref{lem:alg_vs_greedy_v_a} with Lemma~\ref{lem:p/(1+p)} we obtain that for any arrival order of the vertices in $\CS{L} \setminus \Rsample$,
$\E{\ALG} \geq (1-p)\E{\Greedy(G[L'])} \geq \OPT \cdot p(1-p)/(1+p)$. \end{proof}

\begin{corollary}\label{cor:greedy_v_a_c_r}
Algorithm~\ref{alg:greedy_v_a} with $p = \sqrt{2} - 1 \approx 0.41$ is order-oblivious $\left(3-2\sqrt{2}\right)\approx \frac{1}{5.83}$-competitive.
\end{corollary}

Finally, we construct a tight example to show that our analysis of Algorithm~\ref{alg:greedy_v_a} is tight for all $p \in [0,1]$. This construction will also be useful to prove that our analysis is tight in the $\AOS p$ model, for all values of $p$.

\begin{theorem}\label{thm:alg_upper_bound}
For $p \in [0,1]$, there is an infinite sequence of bipartite graphs $\CS{G}_1, \CS{G}_2,\dots$ with increasing number of vertices, where the competitive-ratio of Algorithm~\ref{alg:greedy_v_a} on $\CS{G}_{k}$ approaches $p(1-p)/(1+p)$ as $k \rightarrow{\infty}$. Furthermore, this is true for any arrival order of the vertices in $\CS{L} \setminus \RS{L'}$.
\end{theorem}

We first explain the intuition behind the construction. To this end, we also use the equivalent description of Algorithm~\ref{alg:greedy_v_a} as a price-threshold algorithm.
Following our analysis, there are two paths through which Algorithm~\ref{alg:greedy_v_a} may lose compared to the maximum matching. The first one is due to its reliance on the offline greedy solution on the sample (together with each arriving vertex) to compute candidate edges. This restricts the algorithm to produce matchings that are somewhat similar to the greedy matching on the sample. As Lemma~\ref{lem:p/(1+p)} suggests, this leaves the algorithm with a $p/(1+p)$-fraction of the maximum matching. 
The second path is typical to threshold algorithms and occurs when the heaviest potential match to some $r \in \CS{R}$ is in the sample. This makes the threshold of  $r$  high and thereby prevents $r$ from being matched. This component incurs a loss of a factor of $p$, and leaves the algorithm with a $(1-p)$-fraction of the weight, as Lemma~\ref{lem:alg_vs_greedy_v_a} suggests.

To prove the ratio of $p(1-p)/(1+p)$, the algorithm needs to lose these two factors on one instance. In our construction, the left side consist of three parts $\CS{L}_1,\CS{L}_2$ and $\CS{L}_3$. Figure~\ref{fig:hard_instance_1} depicts the three parts of the same graph $\CS{G}_1$ for $p=1/3$ (the first graph from the sequence in the theorem's statement), and Figure~\ref{fig:hard_instance_2} depicts the third graph $\CS{G}_3$. Firstly, black edges are heavier than gray edges, and dashed edges are the lightest. Secondly, thicker edges correspond to heavier weights. Nevertheless, the weights of all edges are arbitrarily close to $1$.
{
\def\lu{1}
\def\lv{2}
\def\k{4}
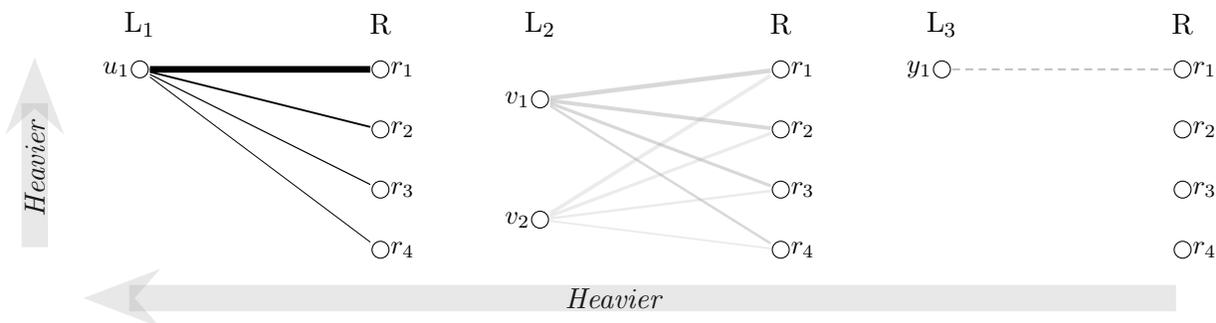
\begin{figure}[t]
\begin{center}
{
\begin{subfigure}[t]{0.03\textwidth}
\begin{tikzpicture}[scale=0.8]
    \node (a) at (0,-0.5) {};
    \node (b) at (0,-\k) {};
    \node[rotate=90] (c) at (0,-2.5) {\textit{Heavier}};
    \draw [-{Stealth[fill=none, scale=0.6]}, line width = 10, solid, opacity=0.18, gray] (b) -> (a);
\end{tikzpicture}
\end{subfigure}
}
{
\begin{subfigure}[t]{0.3\textwidth}
\begin{center}
\begin{tikzpicture}[scale=0.8]
    \node[draw, color=white] at (1,-0.25) {\textcolor{black}{$\text{L}_1$}};
    \node[draw, color=white] at (5,-0.25) {\textcolor{black}{$\text{R}$}};
    
    \foreach \number in {1,...,\lu}
    \draw[] (1,-\number ) node (u-\number) {} node[left, font=\small] {$u_\number$} circle (4pt);
    
    \foreach \number in {1,...,\k}
    \draw[] (5,-\number * 1) node (r-\number) {} node[right, font=\small] {$r_{\number}$} circle (4pt);

    \foreach \number in {1,...,\lu}{
    	\draw [-, solid, line width = (2.5) ] (u-\number) -- (r-\number);
    	\foreach \i in {1,...,\k}{
    		\draw [-, line width = (1.2  - \i / \k), solid, opacity=1] (u-\number) -- (r-\i);
    	}
    }
\end{tikzpicture}
\end{center}
\end{subfigure}
} {
\begin{subfigure}[t]{0.3\textwidth}
\begin{center}
\begin{tikzpicture}[scale=0.8]
    \node[draw, color=white] at (1,-0.25) {\textcolor{black}{$\text{L}_2$}};
    \node[draw, color=white] at (5,-0.25) {\textcolor{black}{$\text{R}$}};
    
    \foreach \number in {1,...,\lv}
    \draw[] (1,-\number * 2 + 0.5) node (v-\number) {} node[left, font=\small] {$v_\number$} circle (4pt);
    
    \foreach \number in {1,...,\k}
    \draw[] (5,-\number * 1) node (r-\number) {} node[right, font=\small] {$r_{\number}$} circle (4pt);

    \foreach \number in {\lv,...,1}{
    	\tikzmath{ \tone = 17 * (\lv - \number + 1);}
    	\foreach \i in {1,...,\k}{
    		\draw [-, line width =  (2.2 - \i / \k - \number/4 ), solid, black!\tone, opacity=0.45] (v-\number) -- (r-\i);
    	}
    }
\end{tikzpicture}
\end{center}
\end{subfigure}
} {
\begin{subfigure}[t]{0.3\textwidth}
\begin{center}
\begin{tikzpicture}[scale=0.8]
    \node[draw, color=white] at (1,-0.25) {\textcolor{black}{$\text{L}_3$}};
    \node[draw, color=white] at (5,-0.25) {\textcolor{black}{$\text{R}$}};
    
    \foreach \number in {1,...,\lu}
    \draw[] (1,-\number * 1) node (y-\number) {} node[left, font=\small] {$y_\number$} circle (4pt);
    
    \foreach \number in {1,...,\k}
    \draw[] (5,-\number * 1) node (r-\number) {} node[right, font=\small] {$r_{\number}$} circle (4pt);

    \foreach \number in {1,...,\lu}{
    	\foreach \i in {1,...,\lu}{
    		\draw [-, line width = (0.9 - \i / 6), densely dashed, lightgray] (y-\number) -- (r-\i);
    	}
    }
\end{tikzpicture}
\end{center}
\end{subfigure}
}
\hfill
{
\begin{subfigure}[t]{0.9\textwidth}
\begin{tikzpicture}[scale=0.8]
    \node (a) at (1,-1) {};
    \node (b) at (19.5,-1) {};
    \node[] (c) at (10,-1) {\textit{Heavier}};
    \draw [-{Stealth[fill=none, scale=0.6]}, line width = 10, solid, opacity=0.18, gray] (b) -> (a);
\end{tikzpicture}
\end{subfigure}
}
\caption{The constructed graph $\CS{G_1}$ for $p = 1/3$}~\label{fig:hard_instance_1}
\end{center}
\end{figure}
\let\lu\relax
\let\lv\relax
\let\k\relax
}
{
\def\lu{3}
\def\lv{6}
\def\k{12}
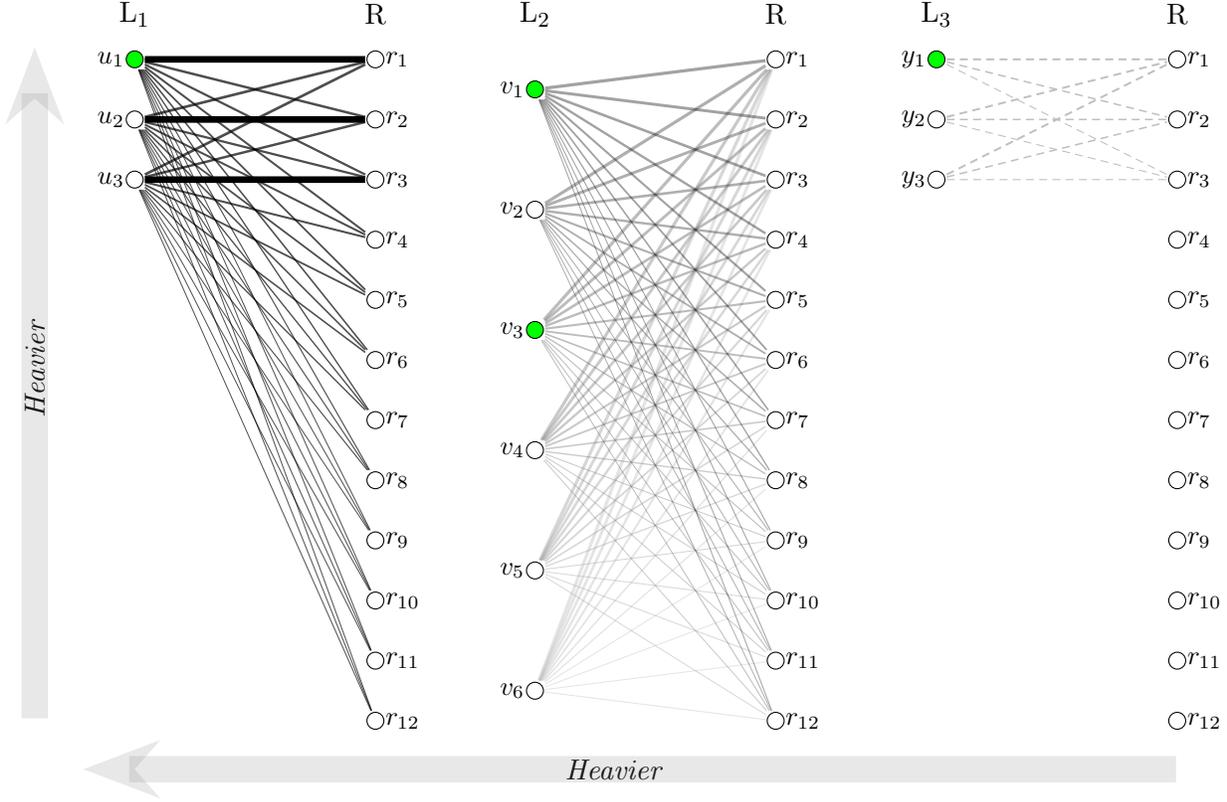
\begin{figure}[t]
\begin{center}
{
\begin{subfigure}[t]{0.03\textwidth}
\begin{tikzpicture}[scale=0.8]
    \node (a) at (0,-0.5) {};
    \node (b) at (0,-\k) {};
    \node[rotate=90] (c) at (0,-\k/2) {\textit{Heavier}};
    \draw [-{Stealth[fill=none, scale=0.6]}, line width = 10, solid, opacity=0.18, gray] (b) -> (a);
\end{tikzpicture}
\end{subfigure}
}
{
\begin{subfigure}[t]{0.3\textwidth}
\begin{center}
\begin{tikzpicture}[scale=0.8]
    \node[draw, color=white] at (1,-0.25) {\textcolor{black}{$\text{L}_1$}};
    \node[draw, color=white] at (5,-0.25) {\textcolor{black}{$\text{R}$}};
    
    \foreach \number in {1,...,\lu} {
        \ifthenelse{\number=1} {
            \draw[fill=green] (1,-\number ) node (u-\number) {} node[left, font=\small] {$u_\number$} circle (4pt);
        } {
            \draw[] (1,-\number ) node (u-\number) {} node[left, font=\small] {$u_\number$} circle (4pt);
        }
    }

    \foreach \number in {1,...,\k}
    \draw[] (5,-\number * 1) node (r-\number) {} node[right, font=\small] {$r_{\number}$} circle (4pt);
    
    \foreach \number in {1,...,\lu}{
    	\draw [-, solid, line width = (2.5) ] (u-\number) -- (r-\number);
    	\foreach \i in {1,...,\k}{
    		\draw [-, line width = (1.1  - \i / 12), solid, opacity=0.7] (u-\number) -- (r-\i);
    	}
    }
\end{tikzpicture}
\end{center}
\end{subfigure}
} {
\begin{subfigure}[t]{0.3\textwidth}
\begin{center}
\begin{tikzpicture}[scale=0.8]
    \node[draw, color=white] at (1,-0.25) {\textcolor{black}{$\text{L}_2$}};
    \node[draw, color=white] at (5,-0.25) {\textcolor{black}{$\text{R}$}};
    
    \foreach \number in {1,...,\lv} {
        \ifthenelse{\number=1 \OR \number=3} {
            \draw[fill=green] (1,-\number * 2 + 0.5) node (v-\number) {} node[left, font=\small] {$v_\number$} circle (4pt);
        } {
            \draw[] (1,-\number * 2 + 0.5) node (v-\number) {} node[left, font=\small] {$v_\number$} circle (4pt);
        }
    }
    \foreach \number in {1,...,\k}
    \draw[] (5,-\number * 1) node (r-\number) {} node[right, font=\small] {$r_{\number}$} circle (4pt);

    \foreach \number in {\lv,...,1}{
    	\tikzmath{ \tone = 13 * (\lv - \number + 2);}
    	\foreach \i in {1,...,\k}{
    		\draw [-, line width =  (1.3 - \i / 10), solid, black!\tone, opacity=0.4] (v-\number) -- (r-\i);
    	}
    }
\end{tikzpicture}
\end{center}
\end{subfigure}
} {
\begin{subfigure}[t]{0.3\textwidth}
\begin{center}
\begin{tikzpicture}[scale=0.8]
    \node[draw, color=white] at (1,-0.25) {\textcolor{black}{$\text{L}_3$}};
    \node[draw, color=white] at (5,-0.25) {\textcolor{black}{$\text{R}$}};
    
    \foreach \number in {1,...,\lu} {
        \ifthenelse{\number=1} {
            \draw[fill=green] (1,-\number * 1) node (y-\number) {} node[left, font=\small] {$y_\number$} circle (4pt);
        } {
            \draw[] (1,-\number * 1) node (y-\number) {} node[left, font=\small] {$y_\number$} circle (4pt);
        }
    
    }
    
    \foreach \number in {1,...,\k}
    \draw[] (5,-\number * 1) node (r-\number) {} node[right, font=\small] {$r_{\number}$} circle (4pt);
    
    \foreach \number in {1,...,\lu}{
    	\foreach \i in {1,...,\lu}{
    		\draw [-, line width = (0.9 - \i / 6), densely dashed, lightgray] (y-\number) -- (r-\i);
    	}
    }
\end{tikzpicture}
\end{center}
\end{subfigure}
}
\hfill
{
\begin{subfigure}[t]{0.9\textwidth}
\begin{tikzpicture}[scale=0.8]
    \node (a) at (1,-1) {};
    \node (b) at (19.5,-1) {};
    \node[] (c) at (10,-1) {\textit{Heavier}};
    \draw [-{Stealth[fill=none, scale=0.6]}, line width = 10, solid, opacity=0.18, gray] (b) -> (a);
\end{tikzpicture}
\end{subfigure}
}
\caption{The constructed graph $\CS{G_3}$ for $p = 1/3$}~\label{fig:hard_instance_2}
\end{center}
\end{figure}
\let\lu\relax
\let\lv\relax
\let\k\relax
}

In general, the edges incident to vertices in $\CS{L}_1$ are heavier than the edges incident to $\CS{L}_{2}$, and the latter are heavier than those incident to $\CS{L_3}$. Additionally, the edges are getting heavier as we move upwards towards vertices with smaller indices. The only exceptions are the edges of the form $(u_i,r_i)$ which are the heaviest in the graph. Both $\CS{L_1}$ and $\CS{L_3}$ consists of $k$ vertices, and $\CS{L_2}$ consists of $k(1-p)/p$ vertices. The right side $\CS{R}$ consists of $k + k/p$ vertices.\footnote{Here, for simplicity, we assume that $1/p$ and $(1-p)/p$ are integers.} Note that $|\CS{L}_1| + |\CS{L}_2| + |\CS{L}_3| = |\CS{R}|$.

In our constructed graphs, there is a perfect matching in which the vertices of $\CS{L_3}$ are matched to the first $k$ vertices on the right, and the vertices of $\CS{L_1}$ and $\CS{L_2}$ (which are connected to all vertices on the right) are matched to the remaining (lower) vertices on the right.

To see how the algorithm performs, first consider the execution of $\Greedy$ on the sample. On average, a $p$-fraction of the vertices from each part appear in the sample. For the first part $\CS{L_1}$, since the edges of the form $(u_i,r_i)$ are the heaviest edges incident to $u_i$ and $r_i$, greedy always matches $u_i$ to $r_i$. Then, since the edges incident to $\CS{L}_2$ are heavier than the edges incident to $\CS{L_3}$, and since the edges incident to higher vertices on the right are heavier, $\Greedy$ matches all the vertices from $\CS{L_2}$ in the sample to the highest available vertices on the right. Hence, the vertices from $\CS{L_1}$ and $\CS{L_2}$ occupy, on average, the first $p |\CS{L_1}| + p|\CS{L_2}| = k$ vertices on the right. So, on average, the vertices from $\CS{L_3}$ in the sample do not have free neighbors to be matched to, and are left unmatched by $\Greedy$. Note that $k = |\CS{R}| \cdot p/(1+p)$. In our depicted example graph $\CS{G_3}$ (Figure~\ref{fig:hard_instance_2}), assume that the sample consists of the following (green) vertices: $u_1,v_1,v_3$ and $y_1$ (which constitutes $1/3$ of the vertices from each part). Then, $\Greedy$ on the sample matches $(u_1,r_1),(v_1,r_2)$ and $(v_3,r_3)$.

Now consider the vertices that arrive after the sampling phase. Note that when $u_i$ is in the sample, $r_i$'s threshold will be set too high, and there will be no candidate edges incident to $r_i$. Hence, in this case, $r_i$ will not be matched by the algorithm (this accounts for the loss of a factor of $p$). In our example, no vertex after the sampling phase will exceed the threshold of $r_1$.

We claim that all candidate edges are incident to vertices in $\CS{R}$ which are matched by $\Greedy$ on the sample, except for one additional vertex -- the highest unmatched vertex in $\CS{R}$. Since $\Greedy$ on the sample matches, on average, only $p/(1+p)$-fraction of the vertices on the right side, and, on average, $p$-fraction of them will not be matched due to their threshold being set too high by the $p\cdot k = |\CS{R}| \cdot p^2 / (1+p) $ $u_i$'s in the sample, we obtain the claimed ratio of $p(1-p)/(1+p)$. The contribution of the one extra matched vertex on the right becomes negligible as the size of the graph grows. 

To see that the claim holds, observe that the candidate edges of the vertices from $\CS{L_1}$ and $\CS{L_3}$ can only be incident to the first $k$ vertices on the right, $r_1,\dots,r_k$ for which $u_i$ is not in the sample (as those set the threshold too high). In our example $u_2,u_3,y_2,y_3$ may only have candidate edges to $r_2$ and $r_3$.
As argued above the first $k$ vertices on the right are matched to the sample.

The candidate edge of $v_i \in \CS{L_2}$ is to the highest $r \in \CS{R}$ such that the weight $(v_i,r)$ exceeds $r$'s threshold. This $r$ is either matched by $\Greedy$ on the sample to some $v_j$ such that $j > i$, or $r$ is the first vertex on the right which is not matched by $\Greedy$ on the sample. Therefore, the claim follows.
In our example, the candidate edges of $v_2$ is $(v_2,r_3)$ and for $v_i \in \set{v_4,v_5,v_6}$ is $(v_i,r_4)$. For a formal proof of Theorem~\ref{thm:alg_upper_bound}, see Appendix~\ref{apx:tight_example}.

\medskip

We conclude this section with an upper-bound on the competitive-ratio of any order-oblivious algorithm for bipartite matching with vertex arrivals. To this end, we adapt an upper-bound 
for the secretary problem in the $\AOS$ model by Kaplan et al.~\cite{kaplan2020competitive}, and obtain an upper-bound of $1/4$ on the competitive-ratio of any order-oblivious algorithm for the special case of the secretary problem. Since there is a $1/e$-competitive algorithm for bipartite matching with vertex arrivals in the random-order model~\cite{DBLP:conf/esa/KesselheimRTV13}, this upper-bound provides a separation between order-oblivious algorithms, and general random-order algorithms for edge-weighted bipartite matching with vertex arrivals.

\begin{theorem}\label{thm:secretary_ob_upperbound}
Any order-oblivious online algorithm for the secretary problem has a competitive-ratio of at most $\frac{1}{4}\left(1 - \frac{5}{n-1}\right) = 1/4 + o(1)$, where $n$ is the size of the online sequence.
\end{theorem}
The proof for Theorem~\ref{thm:secretary_ob_upperbound} is given in Appendix~\ref{proof:secretary_ob_upperbound}.
\section{Vertex Arrivals in Bipartite Graphs in the AOS\texorpdfstring{$p$}{p} Model}\label{sec:bipartite_aos_p}
In this section, we use our order-oblivious analysis from Sections~\ref{sec:vertex_bipartite} to provide a tight analysis for the greedy-based algorithm in the $\AOS p$ model. 
The bipartite matching problem with vertex arrivals in the $\AOS p$ model is defined as the following game between an online player and an adversary:

\begin{enumerate}
\item The adversary generates a bipartite graph $\CS{G}=(\CS{L},\CS{R},\CS{E})$ and a weight function $w:\CS{E}\rightarrow \mathbb{R}_{\geq 0}$.

\item Each vertex $u \in \CS{L}$ is chosen to the sample set $\RS{H}$ independently with probability $p$. $\RS{H}$ is also called the \textit{history set}. $p$, and the induced sub-graph on the vertices $H\cup \CS{R}$ are given to the online player upfront. The \textit{online set}, denoted by $\RS{O}$, is the set of remaining vertices, i.e., $\RS{O}=\CS{L} \setminus \RS{H}$.

\item The vertices in the online set are presented to the online player one-by-one in an adversarial order (that can depend on $\RS{H}$). When a vertex arrives, its incident edges are revealed along with their weights, and the online player has to make an immediate and irrevocable decision whether to match the vertex to an available neighbor, or to leave it unmatched.
\end{enumerate}

An input instance for the problem is defined by $\mathcal{I} = (\CS{G},p)$. The expected performance of the algorithm is compared to the expected total weight of the maximum matching in $\CS{G}[\RS{O}]$ (recall that we write $\CS{G}[\RS{O}]$ instead of  $\CS{G}[\RS{O} \cup \CS{R}]$). More formally, let $\OPT(\cI)$ be the random variable that gets the weight of the maximum matching in $\CS{G}[\RS{O}]$. An algorithm $\ALG$ is called $c$-competitive if for any input instance $\mathcal{I}$, $\E{\ALG(\mathcal{I})}/ \E{\OPT(\mathcal{I})} \geq c$, where the expectation is taken over the random choice of the history set (and the online set), and possibly the internal randomness of the algorithm.

We begin with the case where $p \leq 1/2$, and later show how to generalize the results for $p > 1/2$.\footnote{Algorithm~\ref{alg:aos_p_greedy_v_a} can be viewed as one possible generalization of the optimal algorithm for the secretary problem in the $\AOS$ model \cite{kaplan2020competitive}.}

\IncMargin{1em}
\begin{algorithm}
\caption{AOS$p$ Greedy-Based Vertex Arrivals in Bipartite Graphs ($p \leq 1/2$)}
\label{alg:aos_p_greedy_v_a}
$\RS{L'} \leftarrow \RS{H}$\;
$ \MOnline \leftarrow \emptyset$\;
\For {a vertex $u_{\ell}$ that arrives at round $\ell$ } {
$\RS{G}_\ell \leftarrow \CS{G}[{\RS{L'} \cup \set{u_\ell}}]$\;
    $\RS{M}_\ell \leftarrow \textsc{Greedy}(\RS{G}_\ell)$\;
    \uIf{$u_\ell$ is matched in $\RS{M}_\ell$} {
        Let $(u_\ell, r_\ell) \in \RS{M}_\ell$ be the corresponding edge\tcp*{a candidate edge}
        \If {$r_\ell$ is not matched in $\MOnline$} {
	        $ \MOnline \leftarrow \MOnline \cup \set{(u_\ell, r_\ell)}$\;
        }
    }
}
\Return{\RS{M}}
\end{algorithm}
\DecMargin{1em}

Note that the algorithm does not need to know the sampling probability $p$.

\begin{theorem}\label{thm:p_aos_p(1-p)} For $p \in [0,1/2]$, Algorithm~\ref{alg:aos_p_greedy_v_a} is $p(1-p)$-competitive in the $\AOS p$ model.
\end{theorem}

\begin{proof}
In the $\AOS p$ model, $\RS{H}$ is given to the online player upfront, and the performance of the algorithm is measured compared with $\E{\OPT(\CS{G}[\RS{O}])}$. Each vertex $u \in \CS{L}$ is chosen to $\RS{H}$ independently with probability $p$. Therefore, Algorithm~\ref{alg:aos_p_greedy_v_a} is equivalent to Algorithm~\ref{alg:greedy_v_a} with $\RS{L}' = \RS{H}$. So, by Lemma~\ref{lem:alg_vs_greedy_v_a}, we have $\E{\ALG} \geq (1-p) \E{\Greedy(\CS{G}[H])}$. 

Next, we relate $ \E{\Greedy(\CS{G}[H])}$ with $\E{\OPT(\CS{G}[\RS{O}])}$. Consider the following random process for drawing subsets from $\CS{L}$. For each $u \in \CS{L}$, we draw $x_u \in [0,1]$ uniformly at random. Let $X = \set{u : x_u \leq p}$ and let $Y = \set{u : x_u \leq 1-p}$. Clearly, every vertex is chosen for $X$ independently with probability $p$. Similarly, every vertex is chosen for $Y$ independently with probability $1-p$. Therefore, $\E{\textsc{Greedy}(\CS{G}[\RS{X}])} = \E{\textsc{Greedy}(\CS{G}[\RS{H}])}$ and $\E{\OPT(\CS{G} [\RS{Y}] ) } = \E{\OPT(\CS{G}[\RS{O}])}$. Furthermore, since $p \leq (1-p)$, $\Pr[u \in X \mid u \in Y] = \Pr[x_u \leq p \mid x_u \leq 1- p ] = p /(1-p) $. 

We now apply Lemma~\ref{lem:p/(1+p)} with vertex sampling probability of $q = p/(1-p)$, and obtain that $\E{ \given{\textsc{Greedy}(\CS{G} [\RS{ X }] )}{Y} } \geq  q/(1+q) \cdot \OPT(\CS{G} [\RS{ Y }] )  = p \cdot \OPT(\CS{G} [\RS{ Y }] ) $. By taking the expectation over $Y$ we get that $\E{\textsc{Greedy}(\CS{G}[\RS{X}])} \geq p \cdot \E{\OPT(\CS{G} [\RS{ Y }] )}$. Therefore, we have $\E{\textsc{Greedy}(\CS{G}[\RS{H}])} \geq p \cdot \E{\OPT(\CS{G}[\RS{O}])}$. Overall,
\begin{align*}
    \E{\ALG} &\geq (1-p)\E{\Greedy(\CS{G}[H])} \geq p(1-p)\cdot \E{\OPT(\CS{G}[\RS{O}])}. \qedhere
\end{align*}
\end{proof}
We now show that our analysis is tight.

\begin{theorem}\label{thm:aosp_alg_upper_bound}
For $p\leq 1/2$ there is an infinite sequence of bipartite graphs $\CS{G}_1, \CS{G}_2,\dots$ with increasing number of vertices, where the competitive-ratio of Algorithm~\ref{alg:aos_p_greedy_v_a} on $\CS{G}_{k}$ approaches $p(1-p)$ as $k \rightarrow{\infty}$.
\end{theorem}
The proof of Theorem~\ref{thm:aosp_alg_upper_bound} appears in Appendix~\ref{apx:aosp_tight_example}

\medskip

For $p > 1/2$, using all of $\RS{H}$ as sample (that is, $\RS{L'} \leftarrow \RS{H}$ as in Algorithm~\ref{alg:aos_p_greedy_v_a}) achieves a weaker competitive-ratio than the case where $p=1/2$. Instead, we can randomly discard vertices from $\RS{H}$ and achieve the same competitive-ratio as in the $p=1/2$ case.

\begin{theorem}
For $p > 1/2$, let $\RS{H}' \subseteq \RS{H}$ such that each $u \in \RS{H}$ is drawn to $\RS{H'}$ independently with probability $(1-p)/p$. Then, Algorithm~\ref{alg:aos_p_greedy_v_a} with $\RS{L'} \leftarrow \RS{H'}$ (instead of $L' \leftarrow \RS{H}$) is $1/4$-competitive.
\end{theorem}

\begin{proof}
Observe that for $\Csample \subseteq \CS{L}$,
conditioned on $\RS{O} \cup \RS{L'} = \Csample$, for each $u \in \Csample$ we have $\Pr[u \in \RS{L'} \mid \RS{O} \cup \RS{L'} = \Csample] = 1/2$ independently of all other vertices $u \in \Csample$. Therefore, conditioned on $\RS{O} \cup \RS{L'} = \Csample$ we have an instance of the problem with vertex sampling probability of $1/2$. Hence, we get that $\E{\ALG \mid \RS{O} \cup \RS{L'} = \Csample} \geq\E{\OPT \mid  \RS{O} \cup \RS{L'} = \Csample} / 4$. The competitive-ratio of $1/4$ follows by taking the expectation over $\RS{O} \cup \RS{L}'$.
\end{proof}

An upper-bound of $\min\{p,1/2\}$ on the performance of any online algorithm for bipartite-matching with vertex arrivals can be derived from an upper-bound for the special case of the secretary problem in the $\AOS$ model by Kaplan et al.~\cite{kaplan2020competitive}.

\begin{theorem}\label{thm:sec_min_p_1/2}
For $p \in [0,1]$, any online algorithm for the secretary problem in the $\AOS p$ model has a competitive-ratio of at most $\min \left\{ p + \frac{1}{n-1} \cdot \frac{p^n}{1-p}, \frac{1}{2} \left(1+ \frac{5}{n-1}\right)\right\} = \min \{p, 1/2\} + o(1)$, where $n$ is the total number of elements.
\end{theorem}
The proof of Theorem~\ref{thm:sec_min_p_1/2} is given in Appendix~\ref{proof:thm_sec_min_p_1/2}.
\section{Single Sample Prophet Inequalities and the Two-Faced Model}\label{sec:batched}

Azar et al.~\cite{soda/AzarKW14} introduced a method of obtaining single-sample prophet inequalities from order-oblivious algorithms for online selection problems in the random-order model (which are also called secretary algorithms). We generalize and modify this method to obtain stronger guarantees from order-oblivious algorithms, and specifically from our analyses in Sections~\ref{sec:vertex_bipartite} and~\ref{sec:edge_general}. For future applications, we formulate the method for a wide range of online problems which we call batched online selection problems. Ezra et al.~\cite{ezra2020online} defined a somewhat similar notion of batched prophet inequalities.

\begin{definition}[environment for batched online selection]
An \textit{environment} for a \textit{batched online selection problem} is given by $\cE = (\cU, \cP, \cJ)$, where $\cU = \set{1,\dots,m}$ is the universe of elements, $\cP = \set{P_1,\dots,P_n}$ is a partition of $\cU$ into $n$ disjoint subsets which we call \textit{items}, and $\cJ \subseteq 2^\cU$ is a family of feasible subsets of $\cU$. 
\end{definition}

Given non-negative weights $w= (w_1,\dots,w_m)$ for the elements of $\cU$, the items $P_{j_1},\dots,P_{j_n}$ are presented to an online algorithm one-by-one in some order (which depends on the online model). When an item $P_{j_i}$ arrives, its elements and their weights are revealed to the online algorithm. Upon arrival, the algorithm must decide immediately and irrevocably (before the arrival of the next item) whether to accept or reject each element of the item. The set of accepted elements must be in $\cJ$ at all times. The objective is to maximize the total weight of the accepted elements. Note that online selection problems, as defined in~\cite{soda/AzarKW14}, are a special case of batched online selection problems in which $\cP$ consists of singletons.

We denote an input instance (which also depends on the online model) by $\cI$. For an algorithm $\ALG$, let $\ALG(\cI)$ be the random variable that gets the total weight of the algorithm's final set of accepted elements. Also, let $\OPT(\cI)$ be the maximum weight of a set in $\cJ$ (according to $w$).\footnote{When $w$ is a random variable, $\OPT(\cI)$ is also a random variable.} The objective is to maximize the total weight of $\ALG(\cI)$.

To fully define an online batched selection problem (and an instance for it), we need to specify the online model. The online model determines the way the weights $w$ are selected, the arrival order of the items, and which additional information is revealed to the online algorithm. 

Before we proceed, we give some examples for environments: For $\cJ = \set{\set{i}: i \in \cU}$, and $\cP = \set{\set{i}: i \in \cU}$ we get the classical problem of selecting a single element from a sequence, with the aim of maximizing its weight (also known as the secretary problem in the random-order model). If we take $\cU$ to be the edge set of a bipartite-graph $\CS{G} = (\CS{L},\CS{R},\CS{E})$, and $\cJ$ to be the family of feasible matchings in $\CS{G}$. Then, for $\cP = \set{\set{e : u \in e} : u \in \CS{L}}$ we get the online weighted bipartite-matching problem with (one-sided) vertex arrivals, and for $\cP = \set{\set{e}: e \in \CS{E}}$ we get the online weighted bipartite-matching problem with edge arrivals.

\begin{definition}[random-order batched online selection]\label{def:ro_batched}
In the \textit{random-order} model, given an environment $\cE$ for a batched online selection problem, the weights of the elements $w$ are chosen by an adversary, and the items arrive in a uniformly random order, i.e., $(i_1,\dots,i_n)$ is a uniformly random permutation of $[n]$. In this model, an algorithm is called $c$-competitive, if on any input instance for the problem $\cI = (\cE, w)$, $\E{\ALG(\cI)} \geq c \cdot \OPT(\cI)$, where the expectation is taken over the random arrival order of the items, and the internal randomness of the algorithm. 
\end{definition}

\begin{definition}[single-sample batched prophet inequality]
\label{def:sspi}
In the \textit{single-sample batched prophet inequality}, given an environment $\cE$ for a batched online selection problem, an adversary generates for each item $P_j\in \cP$ a distribution $D_j$ over weight vectors for the elements of $P_j$. The online algorithm gets to draw in advance a single-sample $s=(s_1,\dots,s_n)$ from the product distribution $D = D_1 \times \cdots \times D_n$ for the purpose of learning. Then, the weights of the elements that arrive online, $w=(w_1,\dots,w_m)$, are also drawn from $D$, and the arrival order of the items is chosen by the adversary. In this model, an algorithm for a batched online selection problem is called $c$-competitive, if on any input instance for the problem $\cI = (\cE,D)$, $\E{\ALG(\cI)} \geq c \cdot \E{\OPT(\cI)}$, where the expectation is taken over $s,w \sim D$, and the internal randomness of the algorithm.
\end{definition}

We stress that in contrast to the single-sample prophet inequality~\cite{soda/AzarKW14}, the batched setting allows for correlations between the weights of elements in an item. Observe that   Definition~\ref{def:sspi} generalizes the standard single-sample prophet inequality, in which the weight of each element of $\cU$ is drawn independently from a distribution over its possible weights.

Next, we define a new online model which strengthens the single-sample batched prophet inequality, which we call the two-faced model.
\begin{definition}[two-faced batched online selection]
In the \textit{two-faced} model, given an environment $\cE$ for a batched online selection problem, an adversary chooses for each item $P_i$ a pair of weight vectors for its elements $f_{i,1}, f_{i,2}$. Then, for each $i$, $s_{i} \in \set{f_{i,1}, f_{i,2}}$ is chosen uniformly at random, and let $w_{i} \in \set{f_{i,1}, f_{i,2}} \setminus \set{s_{i}}$ be the remaining vector. The items $(P_1,\dots,P_n)$ weighted by $s = (s_{1},\dots,s_{n})$ are given to the online algorithm upfront as a sample for the purpose of learning. Then, the items weighted by $w = (w_{1},\dots,w_{n})$ arrive online in an adversarial order (that may depend on the realization of $s$ and $w$). Let $f_1 = (f_{1,1},\dots,f_{n,1})$ and $f_2 = (f_{1,2},\dots, f_{n,2})$. In this model, an algorithm for a batched online selection problem is called $c$-competitive, if on any input instance for the problem $\cI = (\cE,f_1,f_2)$, $\E{\ALG(\cI)} \geq c \cdot \E{\OPT(\cI)}$, where the expectation is taken over the random selection of $s$ (and $w$), and the internal randomness of the algorithm.
\end{definition}

Clearly, the two-faced model provides stronger guarantees than single sample batched prophet inequality. Indeed, any $c$-competitive algorithm for a two-faced batched online selection problem is also (at least) $c$-competitive as a single sample batched prophet inequality.\footnote{To show this, consider the following way of drawing the two independent samples $s = (s_1,\dots,s_n)$, $w = (w_1,\dots,w_n)$ from the product distribution $D = D_1 \times \cdots \times D_n$. First, we draw two samples $f_1, f_2 \sim D$ independently, $f_1 = (f_{1,1},\dots,f_{n,1})$, $f_2 = (f_{1,2},\dots,f_{n,2})$. Then, for each $i$, $s_{i} \in \set{f_{i,1}, f_{i,2}}$ is chosen uniformly at random, and let $w_{i} \in \set{f_{i,1}, f_{i,2}} \setminus \set{s_{i}}$ be the remaining vector.

Conditioned on the values of $f_1$ and $f_2$, since the algorithm is $c$-competitive in the two-faced model, we have $\E{\ALG(\cI) \mid f_1,f_2} \geq c \cdot \E{\OPT(\cI) \mid f_1,f_2}$. By taking the expectation over $f_1,f_2 \sim D$, we obtain the competitive-ratio of $c$ for the single sample batched prophet inequality. 
}

Next, we formally define the notion of an order-oblivious algorithm and its competitive-ratio for a random-order batched online selection (see Definition~\ref{def:ro_batched}).

\begin{definition}[order-oblivious~\cite{soda/AzarKW14}]\label{def:order-oblivious}
An online algorithm for a random-order batched online selection problem and its competitive-ratio are called \textit{order-oblivious} if: \begin{enumerate}
    \item \textit{(algorithm)} The algorithm sets an integer $k$ (possibly random), and collects the first $k$ arriving items $P_{i_1},\dots,P_{i_k}$ into a sample (sampling phase). Then, when later items arrive, it uses information from the sample to make accept/reject decisions.
    
    \item \textit{(competitive-ratio)} The competitive-ratio of the algorithm holds for any (possibly adversarial) arrival-order of the items after the sampling phase, $\cP \setminus \set{P_{i_1},\dots,P_{i_k}}$.
    \end{enumerate}
\end{definition}

Now, given an order-oblivious algorithm for a random-order batched online selection problem, we construct an algorithm for the problem in the two-faced model that achieves the same competitive-ratio. The construction is similar to the one in~\cite{soda/AzarKW14}, but works with the items $P_1,\dots,P_n$ instead of single elements.

\begin{algorithm}
\caption{From Order-Oblivious ($\ALG$) to Two-Faced ($\ALG'$)}
\label{alg:two_faced}
\tcp{Offline stage}
\SetKwInOut{Input}{input}
\Input{$\cI$, $P_1,\dots,P_n$ weighted by a sample $s_1,\dots,s_n$.}
Let $k$ be the length of the sampling phase of $\ALG$\; 
Let $j_1,\dots,j_n$ be a uniformly random permutation of $[n]$\;\label{line:j_1}

Pass $P_{j_1},\dots,P_{j_k}$ with the weights  $s_{j_1},\dots,s_{j_k}$ to $\ALG$ as the $k$ first items\;

\tcp{Online stage}
\Input{$P_{i_1},\dots,P_{i_n}$ weighted by $w_{i_1},\dots,w_{i_n}$ arrive online (in adversarial order) }
\For{each arriving $P_{i_j}$ with weights $w_{i_j}$} {
    \uIf{$i_j \in \set{j_1,\dots,j_k}$} {
        \tcp{the index has already been processed as a ``sample''}
        Ignore and continue\;
    } 
    \Else {
        Pass $P_{i_j}$ to $\ALG$ with weight  $w_{i_j}$ and accept/reject in the same way\;
    }
}
\end{algorithm}

\begin{theorem}\label{thm:order_oblivious_two_faced}
Let $\ALG$ be an order-oblivious $c$-competitive algorithm for a random-order batched online selection problem. Then, $\ALG'$ (Algorithm~\ref{alg:two_faced}) is a $c$-competitive algorithm for the problem in the two-faced model.
\end{theorem}

\begin{proof}
Fix $z = (z_1,\dots,z_n)$ such that $z_i \in \set{f_{i,1},f_{i,2}}$ (one arbitrary weight vector from each pair). In our proof we refer only to weight vectors. The corresponding items can be deduced from their index ($z_i$ is a weight vector of the item $P_i$). 

We consider the set of all prefixes of permutations on $[n]$, that is, $\Pi = \{ (\sigma_1,\dots, \sigma_t) : (\sigma_1,\dots,\sigma_n) \in S_n,\ 1 \leq t\leq n \}$.\footnote{$S_n$ is the set of all permutations on $[n]$.} 
Let $J_k = (j_1,\dots,j_k)$ and $J_{n-k} = (j_{k+1},\dots,j_{n})$ (where $j_1,\dots,j_k$ are the first $k$ elements of a random permutation on $[n]$, see Line~\ref{line:j_1} of Algorithm~\ref{alg:two_faced}). For $\pi = (\sigma_1,\dots,\sigma_t) \in \Pi$ and a sequence $x = (x_1,\dots,x_n)$ let $x(\pi) = x(\sigma_1,\dots,\sigma_n) = (x_{\sigma_1},\dots, x_{\sigma_t})$. 

Let $\CS{E}_z$ be the event that the input instance that $\ALG$ gets is weighted by $z=(z_1,\dots,z_n)$, that is, $s(J_k) = z(J_k)$ and $w(J_{n-k}) = z(J_{n-k})$. We have $\Pr[\CS{E}_z] = \sum_{\pi \in \Pi} \Pr\left[ \given{ \CS{E}_z }{ J_k = \pi} \right] \Pr[J_k = \pi]$. Conditioned on $J_k = \pi$, there is exactly one realization of $s = (s_1,\dots,s_n)$ for which $\CS{E}_z$ occurs. Therefore, 
$\Pr[\CS{E}_z]= \sum_{\pi \in \Pi} 2^{-n}\Pr[J_k = \pi] = 2^{-n}$.

We now show that conditioned on $\CS{E}_z$, $s(J_k)$ is distributed as the first $k$ elements in a random permutation of $z_1,\dots,z_n$. For $\pi \in \Pi$, we have
\begin{align*}
    \Pr\left[\given{s(J_k) = z(\pi)} {\CS{E}_z} \right] &= \frac{\Pr[s(J_k) = z(\pi) \land \CS{E}_z]}{\Pr[\CS{E}_z]}\\
    &=  2^n \Pr[s(J_k) = z(\pi) \land s(J_k) = z(J_k) \land w(J_{n-k}) = z(J_{n-k})] \\
    &= 2^n \Pr[J_k = \pi \land s(J_k) = z(J_k) \land w(J_{n-k}) = z(J_{n-k})]  \\
    &= 2^n \Pr\left[ \given{s(J_k) = z(J_k) \land w(J_{n-k}) = z(J_{n-k}) }{ J_k = \pi} \right] \Pr[J_k = \pi] \\ 
    &= 2^n \cdot 2^{-n} \Pr[J_k = \pi] = \Pr[J_k = \pi],
\end{align*}
where the second equality is by the definition of $\CS{E}_z$ and the fact that $\Pr[\CS{E}_z] = 2^{-n}$. The third equality is due to the fact that the events $\set{s(J_k) = z(\pi) \land s(J_k) = z(J_k)}$ and $\set{J_k = \pi \land s(J_k) = z(J_k)}$ are equivalent.

Therefore, conditioned on $\CS{E}_z$, $\ALG$ gets the input instance $\cI_z = (\cE, z)$ and its sampling phase consists of the first $k$ elements in a random permutation of $z$. Since $\ALG$ is a $c$-competitive order-oblivious algorithm, conditioned on $\CS{E}_z$,  the set it selects has expected weight of at least $c \cdot \OPT(\cI_z)$, regardless of the arrival order of the items after the sampling phase. Since $\ALG'$ (Algorithm~\ref{alg:two_faced}) accepts the same set of elements as $\ALG$, we get $\E{\given{\ALG'(\cI)}{ \CS{E}_z}} \geq c\cdot \OPT(\cI_z)$ (where $\cI = (\cE,f_1,f_2)$, and the expectation is taken over the sampling phase $s(J_k)$). By taking the expectation over all possible selections of $z$, we get that $\E{\ALG'(\cI)} \geq c \cdot \E{\OPT(\cI)}$.
\end{proof}

By Theorem~\ref{thm:order_oblivious_two_faced} together with Corollary~\ref{cor:greedy_v_a_c_r} we get the following result in the two-faced model.

\begin{corollary}
There is an algorithm with competitive-ratio of $\left(3-2\sqrt{2}\right) \approx \frac{1}{5.83}$ for the two-faced online bipartite matching with vertex arrivals.
\end{corollary}

We now proceed to show that for bipartite-matching with vertex arrivals (in contrast to the secretary problem) there is a separation between the full knowledge prophet inequality and the two-faced model. We prove an upper-bound of $2/5$ on the competitive-ratio of any algorithm in the two-faced model, whereas Feldman et al.~\cite{feldman2014combinatorial} gave an optimal $1/2$-competitive algorithm in the (full knowledge) prophet inequality setting.

\begin{theorem}
Any algorithm for the two-faced bipartite matching with vertex arrivals has a competitive-ratio at most $2/5$.
\end{theorem}

\begin{proof}
Let $\ALG$ be a $c$-competitive algorithm. We construct several simple instances with right side that consists of two vertices $\CS{R} = \set{r_1,r_2}$, and a slightly different set of online vertices. For an edge $e$, let $f_1(e)$ and $f_2(e)$ be the two possible weights of $e$. Also, let $s(e) \in \set{f_1(e),f_2(e)}$ be the random variable that gets the weight of $e$ in the sample, and let $w(e)$ be the online weight of $e$. 

We consider instances that contain a vertex $z \in \CS{L}$ with zero edge weights ($f_1(z,r) = f_2(z,r) = 0$ for all $r \in \CS{R}$).
Consider an adversary who always reveals $z$ last in the online sequence. We first show that for any such an instance $\cI$, $\ALG(\cI)$ must leave each right side vertex $r \in \CS{R}$ unmatched until the very last online round with probability at least $c$.

To see this, fix $r \in \CS{R}$ and let $\cI$ be an instance as described above. We modify $\cI$ to $\cI'$ by changing the weight $f_1(z,r)$ to an arbitrary large value, so that the weight of any matching in $\cI'$ without $f_1(z,r)$ is negligible  compared to $f_1(z,r)$. More concretely, let $\varepsilon$ be the maximum weight of a matching in $\cI'$ without $f_1(z,r)$. Since the weight $f_1(z,r)$ arrives online with probability $1/2$ (i.e., $w(z,r) = f_1(z,r)$), we have $\E{\OPT(\cI')} \geq f_1(z,r)/2$. On the other hand,
\begin{align*}
\E{\ALG(\cI')} \leq \frac{1}{2} f_1(z,r) \cdot\Pr\left[\given{(z,r) \in \ALG }{ w(z,r) = f_1(z,r)}\right] +  \varepsilon.
\end{align*}
By combining the fact that $\ALG$ is $c$-competitive, we get that
\begin{align*}
    c \leq \frac{\E{\ALG(\cI')}}{\E{\OPT(\cI')}} \leq \Pr\left[\given{(z,r) \in \ALG }{ w(z,r) = f_1(z,r)}\right] + \frac{2\varepsilon}{f_1(z,r)}.
\end{align*}
We can choose $f_1(z,r)$ to be arbitrarily large, therefore, we obtain $\Pr[(z,r) \in \ALG \mid w(z,r) = f_1(z,r)] \geq c$. For $\tilde{\cI} \in \set{\cI, \cI'}$, let $A_{\tilde{\cI},r}$ be the event that $\ALG(\tilde{\cI})$ leaves $r$  unmatched until the last round. Since $z$ arrives last in the sequence, the event $\set{(z,r) \in \ALG }$ implies $A_{\cI',r}$. 
We have $\Pr[ A_{\cI',r} \mid w(z,r) = f_1(z,r)]  \geq c$.

Until reaching the last round, the algorithm cannot distinguish between getting the input $\cI'$ conditioned on $\set{ w(z,r) = f_1(z,r)}$, and getting the input $\cI$ (unconditioned). This is because in both cases the algorithm observes the exact same information. Therefore, 
$\Pr\left[ A_{\cI,r}\right] = \Pr\left[ \given{A_{\cI',r}}{ w(z,r) = f_1(z,r)}\right] \geq c$

We now consider instances $\cI$ with exactly three vertices in $\CS{L}$. To this end, we use the  following notation: For a vertex $u \in \CS{L}$ and $i \in \set{1,2}$, let $f_i(u) = (f_i(u,r_1), f_i(u,r_2))$. The first vertex is $z \in \CS{L}$ as before ($f_1(z) = f_2(z) = (0,0)$). For the second vertex $u_1 \in \CS{L}$, both edges have a weight of $1$ in both faces, i.e., $f_1(u_1) = f_2(u_1) = (1,1)$. The last vertex, $u_2 \in \CS{L}$, takes several forms (each form takes place in a different instance). In each face of $u_2$ ($f_1(u_2)$ and $f_2(u_2)$), it has exactly one edge of weight $\alpha > 1$ (to either $r_1$ or $r_2$) and one edge of weight $0$. That is, $f_1(u_2), f_2(u_2) \in \set{(0,\alpha), (\alpha,0)}$. Note that there are four possible choices for $u_2$ (two choices for $f_1(u_2)$ and two choices for $f_2(u_2)$), but $\set{f_1(u_2) = (0,\alpha),f_2(u_2) = (\alpha,0)}$ and $\set{f_1(u_2) = (\alpha,0),f_2(u_2) = (0,\alpha)}$ lead to the same instance. Therefore, we are left with three different instances.

We consider the arrival order where $u_1$ arrives first, $u_2$ second and $z$ arrives last. The algorithm always observes the same sample weights for $u_1$ and $z$. For $u_2$, let $s_2 = (s(u_1,r_1), s(u_1,r_2))$. The algorithm either observes $s_2 =  (0,\alpha)$ or $s_2 = (\alpha,0)$. 

When $u_1$ arrives online, the algorithm has three possible options: match $(u_1,r_1)$, match $(u_1,r_2)$ or leave $u_1$ unmatched. The decision can be done based only on the sample and the online value of the edges incident to  $u_1$. Since all the values are the same except for $s_2$ which may take two values, the algorithm may have only two different strategies. One for the case $s_2 = (\alpha,0)$ and one for $s_2 = (0,\alpha)$. Let $p_{i}$ be the probability that the algorithm matches $(u_1,r_i)$ conditioned on seeing $s_2 = (\alpha,0)$ and let $p = p_1 + p_2$. Likewise, Let $p'_{i}$ be the probability that the algorithm matches $(u_1,r_i)$ conditioned on seeing $s_2 = (0, \alpha)$ and let $p' = p'_1 + p'_2$. 

From our three possible instances, we choose an instance as follows: If $p_1 \geq p/2$ we take the instance with $f_1(u_2) = f_2(u_2) = (\alpha,0)$. If $p'_2 \geq p'/2$ (and $p_1 < p/2$) we take the instance with $f_1(u_2) = f_2(u_2) = (0,\alpha)$. Otherwise, we have $p_2 \geq p/2$ and $p'_1 \geq p'/2$. In this case we take the instance with $f_1(u_2) = (\alpha,0)$ and $f_2(u_2) = (0,\alpha)$. Roughly speaking, we choose an instance in which whenever $\ALG$ matches $u_1$, it blocks the $\alpha$-weight edge of $u_2$ with probability at least $1/2$.

We begin with the case $p_1 \geq p/2$. Let $q$ be the probability that $\ALG$ adds $(u_2,r_1)$ (the edge of weight $\alpha$) to its matching, conditioned on $r_1$ being available at the second online round. We can write $\E{\ALG} = \E{\ALG(r_1)} + \E{\ALG(r_2)}$ where $\ALG(r_i)$ is the random variable that gets the weight of the edge incident to $r_i$ in the matching produced by $\ALG$. We have $\E{\ALG(r_2)} \leq p_2$ and $\E{\ALG(r_1)} \leq p_1 \cdot 1 + (1-p_1)q \cdot \alpha$. We get that
\begin{align*}
\E{\ALG} \leq (p_1+ p_2)\cdot1  + (1-p_1)  q\cdot \alpha \leq p + (1-p/2) q \cdot \alpha.
\end{align*}
Since $\E{\OPT} \geq 1+\alpha$ and since $\ALG$ is $c$-competitive, it holds that 
\begin{align}
    p + (1-p/2) \cdot q \cdot \alpha &\geq c(1+\alpha) \notag \\ 
    q &\geq \frac{c(1+\alpha) - p}{\alpha(1-p/2)}. \label{eq:q_bound}
\end{align}
Still, the competitive-ratio is upper-bounded by the probability that $r_1$ is not matched before the last round, which is the probability that $u_1$ and $u_2$ are not matched to $r_1$. Thus 
\begin{align*}
    c &\leq 1 - (p_1 + (1-p_1)q) = 1 - p_1(1-q) - q \leq 1-p(1-q)/2 -q \\
    &= 1 - q(1-p/2) -p/2 \leq 1 - \frac{c(1+\alpha) - p}{\alpha} - p/2,
\end{align*}
where the second inequality follows from the fact that $p_1 \geq p/2$ and the last inequality is due to Inequality~\eqref{eq:q_bound}.
We get that 
\begin{align*}
    c(1 + (1+\alpha)/\alpha) &\leq (1 + p/\alpha - p/2) \\
    c &\leq \frac{ (1 + p/\alpha - p/2) }{1 + (1+\alpha)/\alpha}.
\end{align*}
For $\alpha = 2$ we get the upper-bound $c \leq 2 / 5$.

Observe that the second case ($p'_1 \geq p'/2$) is symmetric. So it remains to analyze the third case ($p_2 \geq p/2$ and $p'_1 \geq p'/2$). We have $\E{\ALG} = \frac{1}{2}\E{\ALG \mid s_2= (\alpha,0)} + \frac{1}{2} \E{\ALG \mid s_2 = (0,\alpha)}$. Since 
\begin{align*}
\frac{\E{\ALG}}{\E{\OPT}} = \frac{1}{2}\frac{\E{\ALG \mid s_2= (\alpha,0)}}{\E{\OPT}} + \frac{1}{2}\frac{\E{\ALG \mid s_2= (0,\alpha)}}{\E{\OPT}} \geq c,
\end{align*}
at least one of the terms $\E{\ALG \mid s_2= (\alpha,0)}/ \E{\OPT}$ or $\E{\ALG \mid s_2 = (0,\alpha)} / \E{\OPT}$ must have a value at least $c$. Without loss of generality, $\E{\ALG \mid s_2= (\alpha,0)}/ \E{\OPT} \geq c$. We can now analyze this term in the same way we analyzed the first case ($p_1 \geq p/2$). The only difference is that roles of $r_1$ and $r_2$ are reversed. Hence, we obtain $c \leq 2/5$.
\end{proof}
\section{Edge Arrivals in General Graphs}\label{sec:edge_general}

In the random-order online matching in general graphs with edge arrivals, an adversary chooses a graph $\CS{G} = (\CS{V},\CS{E})$ with non-negative edge weights $w : \CS{E} \rightarrow \mathbb{R}_{\geq 0}$. The cardinality of the edge set $|\CS{E}|$ is revealed to the online algorithm upfront. Then, the edges of $\CS{E}$ arrive one-by-one in a uniformly random order. When an edge $e = (u,v) \in \CS{E}$ arrives, its weight is revealed and the algorithm must either add $e$ to its output matching (subject to the constraint that the set of accepted edges are vertex-disjoint), or reject $e$. The decision is permanent and must be made before the next edge arrives.

In this section, we study the analogue of Algorithm~\ref{alg:greedy_v_a} for edge-arrival in general graphs.
\IncMargin{1em}
\begin{algorithm}
\caption{Order-Oblivious Greedy-Based Edge Arrivals in General Graphs}
\label{alg:greedy_e_a}
$k \leftarrow Binom(|\CS{E}|,p)$\;
Let $\RS{E}' \subseteq \CS{E}$ be the first $k$ edges that arrive online\tcp*{sampling phase}
$ \MOnline \leftarrow \emptyset$\;
\For {an edge $e_{\ell} = (u_\ell,v_\ell)$ that arrives at round $\ell > k$ } {
$\RS{G}_\ell \leftarrow \CS{G}[{\RS{E}' \cup \set{e_\ell}}]$\;
    $\RS{M}_\ell \leftarrow \textsc{Greedy}(\RS{G}_\ell)$\;
    \uIf{$e_\ell \in \RS{M}_\ell$} {
        \tcp{A candidate edge}
        \If {$u_\ell$ and $v_\ell$ are not matched in $\MOnline$} {
	        $ \MOnline \leftarrow \MOnline \cup \set{e_\ell}$\;
        }
    }
}
\Return{\RS{M}}
\end{algorithm}
\DecMargin{1em}
Our analysis of Algorithm~\ref{alg:greedy_e_a} shares many similarities with the analysis of Algorithm~\ref{alg:greedy_v_a} from Section~\ref{sec:vertex_bipartite}.

For a given graph $\CS{G} = (\CS{V},\CS{E})$ consider its directed line graph $\CS{G_D} = (\CS{V_D}, \CS{E_D})$ (see Definition~\ref{def:directed_line_graph}). Let $v_{e_1}, \dots, v_{e_m}$ be the nodes of $\CS{G_D}$ ordered in a non-increasing order of weight. Recall that $v_{e_1}, \dots, v_{e_m}$ is a topological ordering of $\CS{G_D}$. 

We assume that the sampling of $\RS{E'} \subseteq \CS{E}$ is done gradually by coloring the nodes $v_{e_1}, \dots, v_{e_m}$ one-by-one (in this order), each node is colored red independently with probability $p$, and blue otherwise. Then, $\RS{E'}$ is the set of edges $e$ whose node $v_e$ is red.

We define the notion of an \textit{active} node inductively: $v_{e_1}$ is always active. Given the colors of $v_{e_1},\dots,v_{e_{i-1}}$ and their active/inactive status, $v_{e_{i}}$ is active if there is no incoming arc to $v_{e_{i}}$ from an active red node. 

\begin{lemma}\label{lem:red_e_a}
$e_i \in \Greedy({\CS{G}[\REsample \cup \{e_i\}]})$ if and only if $v_{e_i}$ is active.
\end{lemma}

\begin{proof}
We prove this by induction on $i$. Clearly, $e_1 \in \Greedy({\CS{G}[\REsample \cup \{e_1\}]})$ and $v_{e_1}$ is active by definition. Now $e_i = (u,w)$ is added by $\Greedy({\CS{G}[\REsample \cup \{e_i\}]})$ if and only if no heavier edge (in $e_1,\dots,e_{i-1}$) incident to $u$ or $w$ is added by $\Greedy({\CS{G}[\REsample \cup \{e_i\}]})$. By the induction hypothesis, this happens if and only if there is no incoming arc to $v_{e_i}$ from an active red node (the red nodes are $v_{e_j}$ for $e_j \in \REsample$), which by definition means that $v_{e_i}$ is active. 
\end{proof}

Next, we bound the expected performance of the algorithm in terms of the expected performance of $\Greedy$ on the random sample $\REsample$. We first express the expected performance of $\Greedy$ using the probabilities of the nodes $v_{e_1},\dots,v_{e_m}$ to be active.

\begin{lemma}\label{lem:greedy_active_prob}
    $\E{\Greedy(\CS{G}[\REsample])} = p \sum_{i=1}^{m} w(e_i) \Pr[v_{e_i}\text{ is active}].$
\end{lemma}
\begin{proof}
By Lemma~\ref{lem:red_e_a} (together with the fact that $e_i \in \REsample$ if and only if $v_{e_i}$ is red), we have $\Pr[e_i \in \Greedy({\CS{G}[\REsample]})] = \Pr[v_{e_i}\text{ is active and red}]$. Observe that since the colors of $v_{e_1},\dots,v_{e_{i-1}}$ determine whether $v_{e_i}$ is active or not,
the coloring of $v_{e_1},\dots,v_{e_{i}}$
determine if $e_i$ is in the greedy matching.
Furthermore, each node is colored independently, and therefore the events $\set{v_{e_i} \text{ is active}}$ and  $\set{v_{e_i} \text{ is red}}$ are independent, so we have $
    \Pr[e_i \in \Greedy({\CS{G}[\RS{E'}]})] = \Pr[v_{e_i}\text{ is active}] \cdot \Pr[v_{e_i} \text{ is red}] = \Pr[v_{e_i} \text{ is active} ]   \cdot p
$,
and so
\begin{align}
    \E{\Greedy(\CS{G}[\REsample])} &= \sum_{i=1}^{m} w(e_i) \Pr[e_i \in \Greedy({\CS{G}[\REsample]})] =p \sum_{i=1}^{m} w(e_i) \Pr[v_{e_i}\text{ is active}].\qedhere
\end{align}
\end{proof}

We now analyze the performance of the algorithm. We first observe that $e_i$ is a candidate edge of Algorithm~\ref{alg:greedy_e_a} if and only if $v_{e_i}$ is active and blue: For $e_i$ to be a candidate edge we need $e_i \in \CS{E} \setminus \REsample$ and $e_i \in \Greedy({\CS{G}}[\RS{E'} \cup \set{e_i}] )$. $v_{e_i}$ is blue if and only if $e_i \in \CS{E} \setminus \REsample$ and by Lemma~\ref{lem:red_e_a}, $v_{e_i}$ is active if and only if $e_i \in \Greedy({\CS{G}}[\RS{E'} \cup \set{e_i}] )$.

To account for the contribution of $e_i$ to $\MOnline$ we define the notion of a qualifying edge. We say that $e_i$ \textit{qualifies} if $v_{e_i}$ is active, blue, and it has no outgoing arcs to other active blue nodes. When $e_i$ qualifies, we also refer to $v_{e_i}$ as a qualifying node. In our analysis, we take only the contribution of qualifying edges into account. Note that by the definition of $\CS{G_D}$, there is an arc between each pair of nodes $v_{e},v_{e'} \in \CS{V_D}(u)$. Therefore, there is at most one qualifying edge incident to $u$ (i.e., at most one qualifying node in each cluster). Also, if an edge $e$ incident to $u$ qualifies, $v_e$ must be the lowest weight active blue node in $\CS{V_D}(u)$.

By our observations above, when $e_i$ qualifies, it is a candidate edge and there are no candidate edges of smaller weight that intersect $e_i$ (a candidate edge $e_j$ of smaller weight that intersect $e_i$ corresponds to an active blue node $v_{e_j}$ with an arc $v_{e_i} \rightarrow v_{e_j}$). We get that when ${e_i}$ qualifies, it is guaranteed that either $e_i$ or a heavier edge $e_j$ that intersect $e_i$ will be added to $M$. In the latter case, we say that $e_j$ is \textit{stealing} from $e_i$. Note that if $e_j$ is stealing from $e_i$, ${e_j}$ itself is not a qualifying edge because $v_{e_j}$ has an outgoing arc to $v_{e_i}$ (which is active and blue). In addition, observe that $e_j = (u',w')$ might be stealing from at most two qualifying edges (one incident to $u'$ and one incident to $w'$). 

In lemmata~\ref{lem:direct_reimbusement},~\ref{lem:credit_reimbusement} and~\ref{lem:per-vertex-comp}, we consider three different ways to account for the stealing edges. Each way provides us with a different lower bound on the expected performance of the algorithm. 

\begin{lemma}[Direct edge reimbursement]\label{lem:direct_reimbusement} For any arrival order of the edges in $\CS{E} \setminus \REsample$,
$\E{w(\MOnline)} \geq
    \frac{p(1-p)}{2} \cdot  \E{\Greedy(\CS{G}[\REsample])} .$
\end{lemma}

\begin{proof}
When $e_j$ is stealing, we evenly split its weight between the qualifying edges that intersect $e_j$. This way, when $e_i$ qualifies it contributes at least $w(e_i)/2$ to the weight of
$\MOnline$. We have
\begin{align}
    \E{\ALG} &\geq \frac{1}{2}\sum_{i=1}^{m} w(e_i) \Pr[e_i\text{ qualifies}].\label{eq:alg_exp_1_e_a}
\end{align}

We proceed to lower bound the probability that $e_i = (u,w)$ qualifies. Consider the coloring process of the nodes until $v_{e_i}$ such that $v_{e_i}$ is active and blue ($v_{e_1},\dots,v_{e_i}$ are colored and $v_{e_{i+1}},\dots,v_{e_m}$ are still uncolored). For $e_i$ to qualify, $v_{e_i}$ must not have outgoing arcs to active blue nodes. All the outgoing arcs of $v_{e_i}$ are to uncolored nodes in $\CS{V_D}(u) \cup \CS{V_D}(w)$. If there are no uncolored nodes in $\CS{V_D}(u) \cup \CS{V_D}(w)$, then $e_i$ qualifies. Otherwise, we continue the coloring process until reaching an active node $v_{e_j}$ in $\CS{V_D}(u)$ or $\CS{V_D}(w)$ for the first time. Without loss of generality $v_{e_j} \in \CS{V_D}(u)$. If $v_{e_j}$ is colored red, which happens independently with probability $p$, all future nodes in $\CS{V_D}(u)$ will be inactive (as $v_{e_j}$ has an outgoing arc to each one of them), and there will be no outgoing arcs between $v_{e_i}$ to active blue nodes of smaller weight in $\CS{V_D}(u)$. Conditioned on the event that $v_{e_j}$ is red, if there are no uncolored nodes in $\CS{V_D}(w)$, $e_i$ qualifies. Otherwise, we continue the coloring process until reaching the first active node $v_{e_k}$ in $\CS{V_D}(w)$. Once again, $v_{e_k}$ will be colored red independently with probability $p$, and then $e_i$ will qualify. We get that conditioned on the event that $v_{e_i}$ is active and blue, $e_i$ qualifies with probability at least $p^2$. To sum up,
\begin{align}
\begin{split}
    \Pr[e_i\text{ qualifies}] &=  \Pr\left[v_{e_i}\text{ is active and blue}\right] \cdot \Pr\left[\given{{e_i} \text{ qualifies}}{ v_{e_i} \text{ is active and blue}}\right] \\
    &\geq \Pr[v_{e_i}\text{ is active}] \cdot (1-p) p^2.
\end{split}
\label{eq:qualifies_prob_e_a}
\end{align}
By replacing Inequality~\eqref{eq:qualifies_prob_e_a} in Inequality~\eqref{eq:alg_exp_1_e_a}, and by Lemma~\ref{lem:greedy_active_prob}, we get that
\begin{align}
    \frac{\E{w(\MOnline)}}{\E{\Greedy(\CS{G}[\REsample])}} &\geq
    \frac{ \frac{(1-p) p^2}{2} \sum_{e \in \CS{E}}  w(e)\Pr[v_{e}\text{ is active}] }{p \sum_{e \in \CS{E}} w(e)  \Pr[v_{e}\text{ is active}]} =
    \frac{p(1-p)}{2}. \qedhere
\end{align}
\end{proof}

\begin{lemma}[Credit edge reimbursement]\label{lem:credit_reimbusement} For any arrival order of the edges in $\CS{E} \setminus \REsample$,
$\E{w(\MOnline)} \geq
 \frac{(1-p)(2p-1)}{p} \E{\Greedy(\CS{G}[\REsample])}$.
\end{lemma}

\begin{proof}
First recall that for $e_j = (u,w)$, $v_{e_j}$ has outgoing arcs only to nodes in $\CS{V_D}(u) \cup \CS{V_D}(w)$, and that there is at most one qualifying node in each cluster. Thus, $v_{e_j}$ may have outgoing arcs to at most two qualifying nodes: one in $\CS{V_D}(u) \setminus \CS{V_D}(w)$ and one in $\CS{V_D}(w) \setminus \CS{V_D}(u)$. We say that $e_j = (u,w)$ is \textit{penalized} if $v_{e_j}$ is active and blue, and it has outgoing arcs to two distinct qualifying nodes $v_{e_i} \in \CS{V_D}(u) \setminus \CS{V_D}(w)$ and $v_{e_k} \in \CS{V_D}(w) \setminus \CS{V_D}(u)$.

We associate each edge $e_i$ with a budget. If $e_i$ qualifies we deposit a value of $w(e_i)$ to its budget, and if $e_i$ is penalized we charge $w(e_i)$ from its budget.
We now argue that the total amount of weight we deposit minus the total amount of weight we charge is no more than the weight of the output matching.

Fix a coloring of the nodes and the output matching $\MOnline$. We need to show that 
\begin{align}
\sum_{i=1}^{m} w(e_i) \cdot \mathds{1}_{\{e_i\text{ qualifies}\}} \leq w(M) + \sum_{i=1}^{m}  w(e_i)\cdot \mathds{1}_{\{e_i\text{ is penalized}\}}.\label{eq:redistribute}
\end{align}
To this end, we show a way to redistribute the weights on the left hand side of Inequality~\eqref{eq:redistribute}, i.e., $w(M) + \sum_{i=1}^{m}  w(e_i)\cdot \mathds{1}_{\{e_i\text{ is penalized}\}}$, and assign them to the qualifying edges so that each qualifying edge $e_i$ gets at least $w(e_i)$. We redistribute the weights as follows: we split the weight of each stealing edge $e_j \in M$ and its penalty, $w(e_j) \cdot \mathds{1}_{\{e_i\text{ is penalized}\}}$, equally between the qualifying edges that $e_j$ steals from.

Now, consider a qualifying edge $e_i$. As discussed before, when $e_i$ qualifies, either $e_i \in M$, or a non-qualifying heavier edge $e_j$ that intersects $e_i$ is stealing $e_i$ (i.e., $e_j \in M$ instead of $e_i$). Also, recall that stealing edges do not qualify, and therefore the weight of qualifying edges is not redistributed.

If $e_i \in M$, it gets a weight of $w(e_i)$ from $w(M)$. If $e_i \notin M$, then there is a non-qualifying edge $e_j$ that steals $e_i$. Recall that $e_j$ may steal at most two qualifying edges, hence either $e_j$ steals only $e_i$, or $e_j$ steals another qualifying edge $e_k \neq e_i$. In the former case, $e_i$ gets the weight of $e_j$, and $w(e_j) \geq w(e_i)$. In the latter case, $e_j$ is penalized. Hence, a weight of $2 w(e_j)$ (its weight in $M$ and its penalty) is redistributed from $e_j$ and equally split between $e_i$ and $e_k$. Thus, $e_i$ gets a weight of $w(e_j) \geq w(e_i)$.

Overall, we get that $w(M) \geq \sum_{i=1}^{m} w(e_i) \cdot \mathds{1}_{\{e_i\text{ qualifies}\}} - w(e_i)\cdot \mathds{1}_{\{e_i\text{ is penalized}\}}$, and so
\begin{align}
    \E{w(\MOnline)} \geq \sum_{i=1}^{m} w(e_i) \Pr[e_i\text{ qualifies}] - w(e_i)\Pr[e_i\text{ is penalized}].\label{eq:alg_exp_2_e_a}
\end{align}

We now upper bound the probability that $e_j = (u,w)$ is penalized. Consider the coloring process of the nodes until reaching $v_{e_j}$ such that $v_{e_j}$ is blue and active. Now let the coloring process continue until reaching the first active node $v_{e_{t}}$ in $\CS{V_D}(u) \triangle \CS{V_D}(w)$.\footnote{$\CS{V_D}(u) \triangle \CS{V_D}(w) = (\CS{V_D}(u) \cup \CS{V_D}(w))  \setminus \left(\CS{V_D}(u) \cap \CS{V_D}(w) \right) $ is the symmetric difference of the sets.} If there is no such node, then there is at most one qualifying node in $\CS{V_D}(u) \cup \CS{V_D}(w)$, and thus ${e_j}$ is not penalized.\footnote{In this case, only an edge parallel to $e_j$ may qualify.} Without loss of generality, we assume that $v_{e_{t}} \in \CS{V_D}(u) \setminus \CS{V_D}(w)$. To have a qualifying node in $\CS{V_D}(u)$, $v_{e_{t}}$ must be colored blue, which happens independently with probability $(1-p)$. 

We further continue the coloring process until reaching an active node $v_{e_{s}} \in \CS{V_D}(u) \setminus \CS{V_D}(w)$. Once again, if there is no such node, there may be at most one qualifying node in $\CS{V_D}(u) \cup \CS{V_D}(w)$, and thus $e_j$ is not penalized. To have a qualifying node in $\CS{V_D}(w)$, $v_{e_s}$ must be colored blue, which happens independently with probability $(1-p)$.

To conclude, we get that conditioned on the event that $v_{e_j}$ is active and blue, it is penalized with probability at most $(1-p)^2$. Thus,
\begin{align}
\begin{split}
    \Pr[e_j\text{ is penalized}] &=  \Pr[v_{e_j}\text{ is active and blue}] \cdot \Pr[{e_j} \text{ is penalized} \mid v_{e_j} \text{ is active and blue}] \\
    &\leq \Pr[v_{e_i}\text{ is active}] \cdot (1-p)^3.
\end{split}
\label{eq:penalzied_prob}
\end{align}
By replacing Inequality~\eqref{eq:penalzied_prob} in~\eqref{eq:alg_exp_2_e_a}, and using Inequality~\eqref{eq:qualifies_prob_e_a} from Lemma~\ref{lem:direct_reimbusement}, we get that
\begin{align}
    \frac{\E{w(\MOnline)}}{\E{\Greedy(\CS{G}[\REsample])}} &\geq
    \frac{\left((1-p)p^2 - (1-p)^3\right) \sum_{e \in \CS{E}}  w(e)\Pr[v_{e}\text{ is active}] }{p \sum_{e \in \CS{E}} w(e)  \Pr[v_{e}\text{ is active}]} =
    \frac{(1-p)(2p-1)}{p}. \qedhere
\end{align}
\end{proof}

\begin{lemma}[Vertex reimbursement]\label{lem:per-vertex-comp} For any arrival order of the edges in $\CS{E} \setminus \REsample$,
$\E{w(\MOnline)} \geq p^2(1-p) \E{\Greedy(\CS{G}[\REsample])}$.
\end{lemma}

\begin{proof}
For $\Greedy$, we evenly split the weight of each edge in the matching between its two endpoints: For a vertex $u \in \CS{V}$, let $g_u$ be the random variable that gets half the weight of the edge incident to $u$ in $\Greedy(\CS{G}[\REsample])$ (and $0$ if there is no such edge). We also define the random variable $a_u$ for the weight gained by the algorithm from the vertex $u$. $a_u$ gets a non-zero value only when the heaviest active node in $\CS{V_D}(u)$ qualifies. We now upper bound $g_u$ and lower bound $a_u$.

Fix a vertex $u$ and let $w_{u}$ be the random variable that gets the weight of the heaviest active node $v_{e_i} \in \CS{V_D}(u)$ (which is also the first active node from $\CS{V_D}(u)$ in the sequence $v_{e_1},\dots,v_{e_m}$). Only the colors of $v_{e_1},\dots,v_{e_{i-1}}$ determine whether $v_{e_i}$ is the heaviest active node in $\CS{V_D}(u)$. Conditioned on $v_{e_i}$ being the heaviest active node in $\CS{V_D}(u)$, since greedy matches to $u$ the edge with the heaviest active red node, we have $g_u \leq w(e_i) /2$. Therefore, $\E{g_u} \leq \E{w_u / 2}$.

We now define $a_u$. We argued before (Inequaliy~\eqref{eq:qualifies_prob_e_a}) that $e_i$ qualifies with probability at least $p^2(1-p)$. We show that when it qualifies, we can associate a weight of at least $w(e_i)/2$ with $u$. When $e_i$ qualifies it guarantees that either $e_i = (u,w)$ is added to $\MOnline$ or a heavier non-qualifying edge $e_j$ steals $e_i$. In the former case, only $u$ and $w$ are involved and even if $v_{e_i}$ is also the heaviest active node in $\CS{V_D}(w)$, we can evenly split the weight of $e_i$ between $a_u$ and $a_w$. In the latter case, since $v_{e_i}$ is the heaviest active node in $\CS{V_D}(u)$, $v_{e_j}$ cannot be in $\CS{V_D}(u)$. Therefore, $v_{e_j} \in \CS{V_D}(w)$ (as a stealing edge must intersect $e_i$). So $e_j = (w,x)$ and it may steal from another qualifying edge $e_k = (x,y)$ which may be the heaviest active edge in $\CS{V_D}(y)$. Hence, we can evenly split the weight of $e_j$ between $a_u$ and $a_y$. In any case, when $e_i$ qualifies, $a_u \geq w(e_i)/2$. Therefore $\E{a_u} \geq p^2(1-p) \cdot \E{ w_u / 2}$. 

To sum up, we have
\begin{align}
    \frac{\E{w(\MOnline)}}{\E{\Greedy(\CS{G}[\REsample])}} &\geq
    \frac{\sum_{u \in \CS{V}}{ \E{a_u} }}
    {\sum_{u \in \CS{V}}{ \E{g_u} }}
    \geq
    \frac{\sum_{u \in \CS{V}}{ p^2(1-p)\E{ \frac{w_u}{2}} }}
    {\sum_{u \in \CS{V}}{ \E{  \frac{w_u}{2} } }} = p^2(1-p). \qedhere
\end{align}
\end{proof}

By Lemma~\ref{lem:direct_reimbusement}, Lemma~\ref{lem:credit_reimbusement} and Lemma~\ref{lem:per-vertex-comp}, we get the following corollary.
\begin{corollary} \label{cor:edge_against_greedy}
For any arrival order of the edges in $\CS{E} \setminus \REsample$
\begin{align*}
    \E{\ALG} \geq \max\left\{\frac{p(1-p)}{2}, p^2(1-p), \frac{(1-p)(2p-1)}{p} \right\} \E{\Greedy(\CS{G}[\REsample])}.
\end{align*}
\end{corollary}

Next, we establish a tight bound on the expected approximation ratio of $\Greedy$ on a random sample of edges, which is tight even for bipartite-graphs.

\begin{lemma}\label{lem:min_p_1/2}
Let $\RS{E}' \subseteq \CS{E}$ such that each $e \in \CS{E}$ is in $\RS{E}'$ independently with probability $p$. Then $\E{\Greedy(\CS{G}[\RS{E}'])} \geq \min\set{p,1/2} \OPT$.  Moreover, there is a sequence of (bipartite) graphs $\CS{G}_1,\CS{G}_2, \dots$ on which $\E{\Greedy(\CS{G}_k[\RS{E}'])}/ \OPT$ approaches $\min\set{p,1/2}$ as $k \rightarrow{\infty}$.
\end{lemma}

\begin{proof}
For a matching $M$ we evenly split the weight of each edge $e = (u,v) \in M$ between its two endpoints. Let $c(M,u)$ be half the weight of the edge incident to $u$ in $M$ ($c(M,u) = 0$ if there is no such edge). We have $w(M) = \sum_{u \in \CS{V}} c(M,u)$. For convenience of notation, let $M_{E'} = \Greedy(\CS{G}[E'])$

Let $M^*$ be a maximum matching in $\CS{G}$. Fix $e =(u,v) \in M^*$, and let $\CS{F}_e$ be the set of edges incident to $u$ or $v$ with weight at least $w(e)$. That is, $\CS{F}_e = \set{e' : u \in e' \lor v \in e',w(e') \geq w(e) }$. Let $q$ be the probability that an edge from $\CS{F}_e$ is in the matching $M_{E'}$. 

If there is an edge of $\CS{F}_e$ in $M_{E'}$, then $c(M_{E'}, u) + c(M_{E'}, v) \geq w(e)/2$.
If none of the edges of $\CS{F}_e$ are in $M_{E'}$, then in case $e \in \RS{E}'$, $e$ will be added to $M_{E'}$ (since this is the heaviest edge incident to $u$ or $v$, so it will be processed first by $\Greedy$). Since $e \in \RS{E}'$ with probability $p$, independently of whether an edge from $\CS{F}_e$ is in $M_{E'}$ or not, we get that
\begin{align*}
    \E{c(M_{E'}, u) + c(M_{E'}, v)} &\geq q  w(e)/2 + (1-q)p  w(e)\\ 
    &\geq  q \min\set{p,1/2}  w(e) + (1-q)\min\set{p,1/2}  w(e) \\ 
    &\geq  \min\set{p,1/2} w(e).
\end{align*}
Overall,
\begin{align*}
    \E{w(M_{E'})} &\geq \sum_{(u,v) \in M^*} \E{c(M_{E'}, u) + c(M_{E'}, v)} \\ 
    &\geq \sum_{(u,v) \in M^*} \min\set{p,1/2} w(u,v) \\
    &\geq \min\set{p,1/2} w(M^*).
\end{align*}

We now show that the ratio of $\min\set{p,1/2}$ is tight. We first show a simple upper-bound of $p$ by a single instance. Consider a graph that consists of a single edge $(u,v)$ of weight $1$. Clearly, $\OPT = 1$ and $\E{\Greedy(\CS{G}[\REsample])} = p$. 

For the upper-bound of $1/2$, we construct an instance that consists of many ``traps'' which are typically misleading for greedy. The graph we construct is essentially unweighted. The weights we define only serve the purpose of tie-breaking. Consider a bipartite-graph $\CS{G}_{k} = (\CS{L}, \CS{R}, \CS{E})$, where $\CS{R} = \set{r_1,r_2,\dots, r_{2k}}$ and there are two types of vertices in the left side $\CS{L} = \CS{U} \cup \CS{Y}$. $\CS{U} = \set{u_1,\dots,u_k}$ are the vertices of the first type. Each $u_i \in \CS{U}$ is connected to all vertices in $\CS{R}$ through edges of weight larger than $1$ but arbitrarily close to $1$. More concretely, $1+\varepsilon > w(u_i,r_1) > \dots > w(u_i,r_{2k}) > 1$. Furthermore, the edges incident to $u_i$ are heavier than the edges incident to $u_{i+1}$. $ \CS{Y} = \set{y_1,\dots,y_k}$ are the vertices of the second type. Each $y \in \CS{Y}$ is connected only to $r_1,\dots,r_k$ with edges of weight smaller than $1$ but arbitrarily close to $1$. 

Clearly, the weight of $\OPT$ is arbitrarily close to $2k$.
$\Greedy$ processes all the edges incident to $u_i$ in $\REsample$ before all the edges incident to $u_{i+1}$ in $\REsample$. When the edges of $u_i$ in $\REsample$ are processed, there are at least $k - (i-1)$ available vertices from $\set{r_1,\dots,r_k}$ as $u_1,\dots,u_{i-1}$ can occupy at most $i-1$ vertices. Hence, $u_i$ is not matched to one of $r_1,\dots,r_k$ with probability at most $(1-p)^{k-(i-1)}$. So, the expected number of vertices from $\CS{U}$ that are matched outside of $\set{r_1,\dots,r_k}$ is at most $\sum_{i=1}^{k} (1-p)^{k-(i-1)} = \sum_{j=1}^{k} (1-p)^{j} < \sum_{j=1}^{\infty} (1-p)^{j} = (1-p)/p$. Since only the vertices of $\CS{U}$ can be matched to $\set{r_{k+1},\dots,r_{2k}}$, we get that the expected weight of the matching produced by $\Greedy$ is at most $k + (1-p)/p$. Thus, the ratio between $\E{\Greedy(\CS{G}_k[\RS{E}'])}$ and $\OPT$ approached $1/2$ as $k$ approaches infinity.
\end{proof}

\begin{theorem}\label{thm:greedy_e_a_c_r}
For $p \in [0,1]$, Algorithm~\ref{alg:greedy_e_a} is order-oblivious $c(p)$-competitive, where 
\begin{align}
    c(p) =
    \begin{cases}
    \frac{p^2(1-p)}{2} & p \leq \frac{\sqrt{5}-1}{2} \\
    \frac{(1-p)(2p-1)}{2p} & \text{otherwise}.
    \end{cases}
\end{align}
\end{theorem}
\begin{proof}
By Corollary~\ref{cor:edge_against_greedy} together with Lemma~\ref{lem:min_p_1/2}, we get that for any arrival order of the edges in $\CS{E} \setminus \REsample$
\begin{align}
    \E{\ALG} \geq \max\left\{\frac{p(1-p)}{2}, p^2(1-p), \frac{(1-p)(2p-1)}{p} \right\} \cdot \min\left\{p,\frac{1}{2}\right\} \cdot \OPT.
\end{align}
For $p \leq 1/2$, the maximum is obtained by $p(1-p)/2$, for $1/2< p \leq (\sqrt{5}-1)/2$, the maximum is obtained by $p^2(1-p)$, and for $p > (\sqrt{5}-1)/2$ the maximum is obtained by $(1-p)(2p-1)/p$. This proves the theorem.
\end{proof}

By Theorem~\ref{thm:greedy_e_a_c_r} with $p = 1/\sqrt{2}$ together with Theorem~\ref{thm:order_oblivious_two_faced} we get the following result in the two-faced model.
\begin{corollary}
There is an algorithm with competitive-ratio of $\left(3/2 - \sqrt{2}\right) \approx \frac{1}{11.66}$ for the two-faced online matching with edge arrivals in general graphs.
\end{corollary}
\subsection{AOS\texorpdfstring{$p$}{p} Edge Arrivals in General Graphs}

We now use our order-oblivious analysis to derive results for the greedy-based algorithm in the $\AOS p$ model. The definition of the online matching problem with edge arrivals in the $\AOS p$ model is analogous to the vertex arrival case (see Section~\ref{sec:bipartite_aos_p}). Also, as in the vertex arrivals case, for $p \leq 1/\sqrt{2}$, we simply use Algorithm~\ref{alg:greedy_e_a} and replace the sampling-phase with the history set $\RS{H}$.

\begin{theorem}\label{thm:p_aos_e_a_competitive}
For $p \leq 1/\sqrt{2}$, Algorithm~\ref{alg:greedy_e_a} with the sampling-phase replaced with $\RS{E'} \leftarrow \RS{H}$ is $c(p)$-competitive in the $\AOS p$ model, where
\begin{align}
    c(p) =
    \begin{cases}
    p^2 / 2 & p \leq 1/3 \\
    p(1-p)/ 4 & 1/3 < p \leq 1/2 \\
    p^2(1-p)/2 & 1/2 < p \leq (\sqrt{5} - 1)/2\\
    (1-p)(2p-1)/(2p) & (\sqrt{5} - 1)/2 < p \leq 1/\sqrt{2}.
    \end{cases}
\end{align}
\end{theorem}

\begin{proof}
Each edge $e \in \CS{E}$ is chosen to $\RS{H}$ independently with probability $p$. Hence, by Corollary~\ref{cor:edge_against_greedy}, we have $\E{\ALG} \geq \max\{(1-p)p/2, (1-p)p^2, (1-p)(2p-1)/p \} \E{\Greedy(\CS{G}[H])}$. 

We now relate $\E{\Greedy(\CS{G}[H])}$ to $\E{\OPT(\CS{G}[\RS{O}])}$. We begin with the case $p \leq 1/2$. For each $e \in \CS{E}$ we draw $x_e \in [0,1]$ uniformly at random. Let $X = \set{e | x_e \leq p}$, and $Y = \set{e | x_e \leq 1-p}$. We have $\E{\textsc{Greedy}(\CS{G}[\RS{X}])} = \E{\textsc{Greedy}(\CS{G}[\RS{H}])}$, $\E{\OPT(\CS{G} [\RS{Y}] ) } = \E{\OPT(\CS{G}[\RS{O}])}$, and $\Pr[u \in X | u \in Y] = p /(1-p) $. 

We now apply Lemma~\ref{lem:min_p_1/2} with edge sampling probability of $q = p/(1-p)$, and obtain that $\E{ \given{\textsc{Greedy}(\CS{G} [\RS{ X }] )}{Y} } \geq  \min\set{p/(1-p), 1/2} \cdot \OPT(\CS{G} [\RS{ Y }] )$. By taking the expectation over $Y$ we get that $\E{\textsc{Greedy}(\CS{G}[\RS{X}])} \geq  \min\set{p/(1-p), 1/2} \cdot \E{\OPT(\CS{G} [\RS{ Y }] )}$. Therefore, we have $\E{\textsc{Greedy}(\CS{G}[\RS{H}])} \geq \min\set{p/(1-p), 1/2} \cdot \E{\OPT(\CS{G}[\RS{O}])}$. Overall,
\begin{align*}
    \E{\ALG} &\geq  \max\left\{\frac{p(1-p)}{2}, p^2(1-p), \frac{(1-p)(2p-1)}{p} \right\} \E{\Greedy(\CS{G}[H])} \\
    &\geq \max\left\{\frac{p(1-p)}{2}, p^2(1-p), \frac{(1-p)(2p-1)}{p} \right\} \cdot \min\set{\frac{p}{1-p}, \frac{1}{2}} \cdot  \E{\OPT(\CS{G}[\RS{O}])}. 
\end{align*}
For $p \leq 1/2$, it holds that $\max\{(1-p)p/2, (1-p)p^2, (1-p)(2p-1)/p \} = (1-p)p/2$.  For $p \leq 1/3$, we have $\min\set{p/(1-p), 1/2} = p/(1-p)$, and thus $\E{\ALG} \geq p^2 \E{\OPT}/2$. For $1/3 < p \leq 1/2$, we have $\min\set{p/(1-p), 1/2} = 1/2$, and so $\E{\ALG} \geq p(1-p) \E{\OPT}/4$.

For $1/2 <  p \leq (\sqrt{5} - 1)/2$, it holds that  $\max\{(1-p)p/2, (1-p)p^2, (1-p)(2p-1)/p \} = (1-p)p^2$, and $\min\set{p/(1-p), 1/2} = 1/2$. Thus, $\E{\ALG} \geq p^2(1-p) \E{\OPT}/2$.

For $(\sqrt{5} - 1)/2 < p \leq 1/\sqrt{2}$, it holds that $\max\{(1-p)p/2, (1-p)p^2, (1-p)(2p-1)/p \} = (1-p)(2p-1)/p$ and $\min\set{p/(1-p), 1/2} = 1/2$, so we get $\E{\ALG} \geq (1-p)(2p-1)\E{\OPT}/(2p) $.
\end{proof}

For $p > 1/ \sqrt{2}$, we randomly discard some of the edges from $\RS{H}$ and achieve the same competitive-ratio as in the $p=1/\sqrt{2}$ case.

\begin{theorem}
For $p > 1/\sqrt{2}$, let $\RS{H}' \subseteq \RS{H}$ such that each $e \in \RS{H}$ is drawn to $\RS{H'}$ independently with probability $\left(1+\sqrt{2}\right)(1-p)/p$. Then, Algorithm~\ref{alg:greedy_e_a} with the sampling-phase replaced with $\RS{E'} \leftarrow \RS{H'}$ is $\left(3/2 - \sqrt{2}\right)$-competitive.
\end{theorem}

\begin{proof}
Let $\CEsample \subseteq \CS{E}$.
Conditioned on $\RS{O} \cup \RS{E'} = \CEsample$, for every $e \in \CEsample$ we have $\Pr[e \in \RS{E'} | \RS{O} \cup \RS{E'} = \CEsample] = 1/\sqrt{2}$ independently of all other edges. Therefore, conditioned on $\RS{O} \cup \RS{E'} = \CEsample$ we have an instance of the problem with edge sampling probability of $1/\sqrt{2}$. Hence, we get that $\E{\ALG | \RS{O} \cup \RS{E'} = \CEsample} \geq \left(3/2 - \sqrt{2}\right){\E{\OPT |  \RS{O} \cup \RS{E'} = \CEsample}} $. By taking the expectation over $\RS{O} \cup \RS{E}'$, we get $\E{\ALG} \geq \left(3/2 - \sqrt{2}\right)\E{\OPT}$.
\end{proof}

\section{Discussion}\label{sec:discussion}
In this paper we study weighted matching problems in sample-based adversarial-order online models. At the base of our method lies an order-oblivious analysis which we use to derive results in both the recent adversarial-order model with a sample and our new two-faced model which strengthens the single-sample prophet inequality setting. Through this method, we unveil more relations between different online models, and provide a comprehensive picture of the performance of the greedy-based  algorithms which we study.

While our analysis of the greedy-based algorithm for weighted bipartite matching is tight, we believe that this is not the case for edge arrivals in general graphs. Finding the exact competitive-ratio in this case is an interesting open question for future research. Beyond the greedy-based algorithm, it would be interesting to close the remaining gaps between the lower and upper bounds
on the competitive ratio for the matching problems which we study.
Inspired by the relation to the random-order model, and specifically by the optimal algorithm for weighted bipartite matching with vertex arrivals in the random-order model~\cite{DBLP:conf/esa/KesselheimRTV13}, one could hope to improve upon the greedy-based algorithm by replacing the offline greedy algorithm with the offline maximum matching. In Appendix~\ref{apx:opt_based}, we show that this approach leads to a non-competitive algorithm. Hence, in case our lower bound is not tight, a different algorithmic approach should be designed to improve it.

Another exciting direction is to study unweighted matching problems in sample-based adversarial-order models. In particular, it would be interesting to discover whether the $(1-1/e)$ barrier for unweighted bipartite matching with vertex arrivals in the standard worst-case model can be surpassed by using a sample.

\bibliographystyle{plain}
\bibliography{citations}

\appendix
\section{Tight Examples}

\subsection{Proof of Theorem~\ref{thm:alg_upper_bound}}\label{apx:tight_example}
\begin{proofw}
Fix $p \in [0,1]$. For an integer $k \geq 1$, we construct $\CS{G}_k =(\CS{L},\CS{R},\CS{E})$ with $\CS{R}=\set{r_1,\dots,r_{k + k/p}}$ and $\CS{L}= \CS{L_1}\cup \CS{L_2} \cup \CS{L_3}$ where $\CS{L_1}=\set{u_1,\dots,u_k}$, $\CS{L_2} = \set{v_1,\dots, v_{k(1-p)/p}}$ and $\CS{L_3} = \set{y_1,\dots,y_{k}}$.\footnote{we assume that $1/p$ and $(1-p)/p$ are integers. When this is not the case, we can round down the terms $k(1-p)/p$ and $k/p$.} Each vertex in $\CS{L_1} \cup \CS{L_2}$ is connected to all vertices in $\CS{R}$. On the other hand, the vertices of $\CS{L_3}$ are connected only to the first $k$ vertices in $\CS{R}$, i.e., $r_1,\dots,r_k$.

The graph is essentially unweighted: All edge weights are between $1$ and $1+\varepsilon$ for an arbitrarily small $\varepsilon$. Hence, we treat all weights as $1$ when accounting for their contribution to the weight of the matching. The different weights are used to define an order on the edges, which determines the tie-breaking rules against the algorithm.

The heaviest edges in the graph are the $k$ edges $\{(u_1,r_1),(u_2,r_2), \dots,(u_k,r_k)\}$ (the order among these edges may be arbitrary).
All remaining edge weights are ordered as follows: Firstly by the index of the subset to which their left side vertex belongs ($\CS{L}_1,\CS{L}_2$ or $\CS{L}_3$) in decreasing order. Secondly, by the index of their left side vertex in decreasing order, and finally by the index of their right side vertex in decreasing order.

More formally, let $e_1 = (x_{i_1}, r_{j_1})$, $e_2 = (x_{i_2}, r_{j_2})$, and let $\ell_1,\ell_2 \in \set{1,2,3}$ be the indices of the subsets of $\CS{L} = \CS{L_1}\cup \CS{L_2} \cup \CS{L_3}$ such that $x_{i_1} \in \CS{L}_{\ell_1}$ and $x_{i_2} \in \CS{L}_{\ell_2}$. Then, $w(e_1) > w(e_2)$ if $e_1 \in \{(u_1,r_1), \dots,(u_k,r_k)\}$ and $e_2 \notin \{(u_1,r_1), \dots,(u_k,r_k)\} $ or if
$(\ell_1,i_1,j_1) < (\ell_2,i_2,j_2)$ lexicographically.

To lower bound $\OPT$, observe that there is a perfect matching in $\CS{G}_k$, and therefore $\OPT \geq k+k/p$.
To upper bound $\E{\ALG}$, we upper bound the number of vertices $r \in \CS{R}$ such that there is a candidate edge incident to $r$.  Let $\RS{O} = \CS{L} \setminus \RS{L}'$ be the set of vertices that arrive after the sampling phase. Let $s_i$ and $o_i$ be the random variables that gets the number of vertices in $\CS{L}_i \cap \RS{L'}$ and $\CS{L}_i \cap \RS{O}$, respectively.

Consider the matching $\Greedy(\CS{G}[L'])$. In this matching, greedy first matches the edges of the form $(u_i,r_i)$, then it proceeds to the vertices in $\CS{L_2} \cap \RS{L'}$ and matches them to the available vertices in $\CS{R}$ with the lowest index. Finally, the vertices of $\CS{L_3} \cap \RS{L'}$ are matched to the remaining available vertices in $\{r_1,\dots,r_k\}$ (if there are any).

First, observe that if $u_i \in \RS{L'}$, then there will be no candidate edges incident to $r_i$. This is because $(u_i,r_i)$ is the heaviest edge incident to both $u_i$ and $r_i$, and therefore it always participates in $\Greedy(\CS{G}[\RS{L'} \cup \set{x_i}])$ (for every vertex $x_i \in \CS{L} $). So, the vertices in $\CS{L_1} \cap \RS{O}$ and $\CS{L_3} \cap \RS{O}$, can only have candidate edges to vertices $r_j \in \set{r_1,\dots,r_k}$ such that $u_j \in \RS{O}$.

For each $v_i \in \CS{L_2} \cap \RS{O}$, in the matching $\Greedy(\CS{G}[\RS{L'} \cup \set{v_i}])$, $v_i$ can only be matched to one of $\{r_1,\dots,r_{s_1+s_2 +1} \}$. This is because there are only $s_1 + s_2 + 1$ vertices from $\CS{L_1} \cup \CS{L_2}$ in $\RS{L'} \cup \set{v_i}$, and greedy matches them to the vertices in $\CS{R}$ with the lowest indices. 
Therefore, we get that the candidate edge for $v_i$, $(v_i,r_j)$ is either to $r_j \in \set{r_1,\dots,r_k}$ such that $u_j \in \RS{O}$, or, in case $s_1+s_2 \geq k$, it may also be that $r_j \in \{ r_{k +  1},\dots, r_{s_1+s_2+1} \}$. 

To sum up, the vertices in $\CS{R}$ that may have an incident candidate edge are $\{r_1,\dots,r_k : u_j \in O\}$ and $\{r_{k+1},\dots,r_{s_1+s_2+1}\}$. We have $|\{r_1,\dots,r_k : u_j \in O \} | = o_1$ and  $|\{r_{k+1},\dots,r_{s_1+s_2+1}\}| = \max\{s_1 + s_2 + 1 - k, 0\}$. Hence, at most $\max\{s_1 + s_2 + 1 - k, 0\} + o_1 $ vertices in $\CS{R}$ are matched by the algorithm. Since $s_1 + o_1 = |\CS{L_1}| = k$, we have $\max\{s_1 + s_2 + 1- k, 0\} + o_1  = \max\{s_2+1, o_1\} \leq \max\{s_2,o_1\} + 1$. 

It remains to upper-bound $\E{\max\set{s_2, o_1}}$. We have $\E{o_1} = \E{s_2} = (1-p)k$ as every vertex is drawn to $\RS{L'}$ independently with probability $p$. We also use the trivial upper-bound of $|\CS{R}| = k + k/p$ on the size of the matching. For $0< \delta <1$, we have 
\begin{align}
\begin{split}
    \E{\max\set{s_2, o_1} } &\leq (1+\delta)(1-p)k \Pr\left[\max\set{s_2, o_1} \leq (1+\delta)(1-p)k \right] \\
    &\quad +  (k+k/p) \Pr\left[\max\set{s_2, o_1} > (1+\delta)(1-p)k  \right].~\label{eq:o_o_max_exp}
\end{split}
\end{align}
By Chernoff bound we get that
$\Pr\left[o_1 >  (1+\delta)(1-p)k  \right] \leq e^{-\delta^2 (1-p)k/3}$, and for $\delta = \sqrt{3\log{k}/(k (1-p))}$, we obtain $\Pr\left[o_1 \geq  (1+\delta)(1-p)k  \right] \leq 1/k$. Similarly, we get the same upper-bound for $\Pr[s_2 >  (1+\delta)(1-p)k ]$. Hence, by applying a union bound we get that $\Pr\left[\max\set{s_2, o_1} > (1+\delta)(1-p)k  \right] \leq 2/k$.  By substituting the last inequality in~\eqref{eq:o_o_max_exp} and using the trivial upper-bound of $1$ on the first probability in~\eqref{eq:o_o_max_exp}, we get that
\begin{align}
    \E{\ALG} \leq \E{\max\set{o_1,s_2}} + 1 \leq \left(1+\delta \right)(1-p)k +  2(1+1/p) + 1.\label{eq:o_o_ALG_performance}
\end{align}
Together with the fact that $\OPT \geq k(1+1/p)$, we obtain
\begin{align*}
    \frac{\E{\ALG}}{\OPT} &\leq \frac{\left(1+\delta \right)(1-p)k +  2(1+1/p) + 1}{k(1+1/p)} \leq  \left(1 + \sqrt{\frac{3}{1-p} \cdot \frac{\log{k}}{k}}\right) \frac{p(1-p)}{1+p} + \frac{3}{k},
\end{align*}
which approaches $p(1-p)/(1+p)$ as $k \rightarrow{\infty}$.
\end{proofw}

\subsection{Proof of Theorem~\ref{thm:aosp_alg_upper_bound}}\label{apx:aosp_tight_example}
\begin{proofw}
We use the construction of $\CS{G}_k$ as in the proof of Theorem~\ref{thm:alg_upper_bound} with a slight modification. It is easy to verify that our upper-bound on $\E{\ALG}$ by Inequality~\eqref{eq:o_o_ALG_performance} holds for any size of $\CS{L_3}$.
In the maximum matching on the entire graph the vertices of $\CS{L_3}$ are matched to  $\set{r_1,\dots,r_k}$. Since now we compare the performance of the algorithm to $\E{\OPT(\CS{G}[O])}$, we add more vertices to $\CS{L_3}$ so that with high probability, at least $k$ vertices from $\CS{L_3}$ will be in $\RS{O}$, and $\OPT(\CS{G}[O])$ will be able to use them to match all the vertices in $\set{r_1,\dots,r_k}$.

More concretely, we take $\CS{L_3} = \set{y_1,\dots,y_{k^2/(1-p)}}$ instead of $\set{y_1,\dots,y_k}$.
Let $o_i$ and $h_i$ be the random variables that get the number of vertices in $\CS{L}_i \cap \RS{O}$ and $\CS{L}_i \cap \RS{H}$, respectively. Since all vertices in $\left(\CS{L_1} \cup \CS{L_2}\right) \cap \RS{O}$ can be matched to vertices in $\set{r_{k + 1},\dots, r_{k+k/p}}$ (as $|\CS{L_1}\cup \CS{L_2}| = k/p$), and the vertices in $\set{r_1,\dots,r_k}$ can be matched to vertices in $\CS{L_3} \cap \RS{O}$, we have $\OPT(\CS{G}[O]) \geq \min\set{k,o_3} + o_1 + o_2$. Observe that $\E{o_1} = (1-p)k$, $\E{o_2} = k(1-p)^2/p$ and $\E{o_3} = k^2$ as every vertex is drawn to $\RS{O}$ independently with probability $1-p$. 
By applying a Chernoff bound, we obtain $\Pr[o_3 < k] \leq e^{-(k-1)^2/2}$ and thus
\begin{align*}
     \E{\OPT(\CS{G}[O])} \geq (1- e^{-(k-1)^2/2})k + (1-p)k + k(1-p)^2/p = k(1-e^{-(k-1)^2/2})/p.
\end{align*}
Together with Inequality \eqref{eq:o_o_ALG_performance}, we obtain
\begin{align*}
    \frac{\E{\ALG}}{\E{\OPT}} &\leq \frac{\left(1 + \sqrt{\frac{3}{1-p} \cdot \frac{\log{k}}{k}}\right)(1-p)k + 2(1+1/p) + 1}{k(1-e^{-(k-1)^2/2})/p} =  (1 + o(1)) p(1-p).\qedhere
\end{align*}
\end{proofw}

\section{Upper Bounds for the Secretary Problem}

\subsection{Proof of Theorem~\ref{thm:secretary_ob_upperbound}}\label{proof:secretary_ob_upperbound}
\begin{proofw}
Let $\ALG$ be an order-oblivious algorithm for the secretary problem. Fix $n \in \mathbb{N}$. $\ALG$ can be viewed as a family of functions $\{P^{i}_{j}\}_{1\leq i <  j  \leq n}$, $P^{i}_{j} : \mathbb{R}^{j}_{\geq 0} \rightarrow [0,1]$, such that $P^{i}_{j}(x_1,\dots,x_j)$ is the probability that $\ALG$ accepts $x_j$ conditioned on setting the length of the sampling phase to $k=i$, reaching round $j > i$ (without accepting) and observing the values $x_1,\dots,x_j$ at the first $j$ online rounds. 
As noted by Kaplan et al.~\cite{kaplan2020competitive}, a result by Moran et al.~\cite{Moran} (Corollary 3.7) implies that there is an infinite set $\cV \subseteq \mathbb{N}$ such that $\ALG$ (i.e., $\{P^{i}_{j}\}_{1\leq i <  j  \leq n}$) on inputs from $\cV$ depends only on ordinal information. Hence, there is an input $\cC = \{\alpha_1,\dots,\alpha_n\} \subseteq \cV$ such that $\alpha_1 > \alpha_2 > \dots > \alpha_n$ and $\alpha_1$ can be arbitrarily larger than $\alpha_2$. 

Consider the following partition of $\{P^{i}_{j}\}_{1\leq i <  j  \leq n}$ to $n+1$ subsets $\{P^{0}_{j}\}_{j > 0}, \{P^{1}_{j}\}_{j > 1}, \dots, \{P^{n-1}_{n}\}$. Each subset $\{P^{i}_{j}\}_{j > i}$ can be viewed as an algorithm for the secretary problem in the $\AOS$ model with history set of size $h = i$ and online set of size $n-i$.\footnote{The elements from the history set can be randomly permuted and fed to $\ALG$ as the first $h$ elements.}
We can therefore follow the result by Kaplan et al.~\cite{kaplan2020competitive} (Theorem 2.2), and get that for $i \geq 1$, $\Pr[\ALG(\cC) \text{ accepts } \alpha_1 \mid k = i] \leq \frac{i}{n} \cdot \frac{n-i}{n-1}$. If $k = 0$, this is exactly the adversarial-order setting with no sample. In this case, it is well known that no online algorithm can accept the largest element with probability greater than $1/(n-1)$ (see~\cite{lecture_kesselheim_16} for example). In other words,  $\Pr[\ALG(\cC) \text{ accepts } \alpha_1 \mid k = 0] \leq \frac{1}{n-1}$. Therefore, we have
\begin{align*}
    \Pr[\ALG(\cC) \text{ accepts } \alpha_1] &\leq \sum_{i = 0}^{n} \Pr[\ALG(\cC) \text{ accepts } \alpha_1 \mid k = i ]\Pr[k = i] \\
    &\leq \sum_{i=1}^{n}\frac{i}{n} \cdot \frac{n-i}{n-1} \Pr[k=i] + \frac{1}{n-1} \Pr[k=0] \leq  \frac{1}{4}\left(1  + \frac{5}{n-1}\right),
\end{align*}
where the last inequality follows from the fact that $\frac{i}{n} \cdot \frac{n - i}{n - 1}$ is maximized for $i = n/2$.

It holds that $\OPT(\cC) = \alpha_1$, and $\E{\ALG(\cC)} \leq \alpha_1 \Pr[\ALG(\cC) \text{ accepts } \alpha_1] + \alpha_2$. Therefore, $\E{\ALG} / \OPT \leq \Pr[\ALG \text{ accepts } \alpha_1] + \alpha_2/ \alpha_1$. Since we can choose $\alpha_1$ to be arbitrarily larger than $\alpha_2$, we get that $\E{\ALG}/ \OPT \leq \frac{1}{4}\left(1  + \frac{5}{n-1}\right)$. 
\end{proofw}

\subsection{Proof of Theorem~\ref{thm:sec_min_p_1/2}}\label{proof:thm_sec_min_p_1/2}
\begin{proofw}
Let $\ALG$ be an algorithm for the secretary problem in the $\AOS p$ model, and fix $n \in \mathbb{N}$ and $p \in [0,1]$. $\ALG$ can be viewed as an order-oblivious algorithm for the secretary problem with $k = Binom(n,p)$. Therefore, following the proof of Theorem~\ref{thm:secretary_ob_upperbound}, there is an input $\cC = \set{\alpha_1,\dots,\alpha_n}$, such that $\alpha_1 > \alpha_2 > \dots > \alpha_n$, $\alpha_1$ is arbitrarily larger than $\alpha_2$, and
\begin{align*}
    \Pr[\ALG \text{ accepts } \alpha_1] &\leq \sum_{i = 1}^{n} \frac{i}{n} \cdot \frac{n-i }{n-1}  \Pr[k = i]  + \frac{1}{n-1} \Pr[k = 0]\\
    &= \frac{1}{n(n-1)}\sum_{i = 0}^{n} i(n-i) \binom{n}{i} p^i (1-p)^{n-i} + \frac{p^n}{n-1}\\
    &= \frac{(1-p)}{n-1}\sum_{i = 0}^{n-1} i\binom{n-1}{i} p^i (1-p)^{n-1-i} + \frac{p^n}{n-1}\\
    &= \frac{(1-p)}{n-1} (n-1)p + \frac{p^n}{n-1} \\
    &= p(1-p) + \frac{p^n}{n-1}.
\end{align*}

Let $\cI = (\cC, p)$ denote the input instance in the $\AOS p$ model. We have $\E{\OPT(\cI)} \geq (1-p) \alpha_1$, and  $\E{\ALG(\cI)} \leq \alpha_1 \Pr[\ALG \text{ accepts } \alpha_1] + \alpha_2$. Since we can choose $\alpha_1$ to be arbitrarily larger than $\alpha_2$, we get the upper-bound of $p + \frac{1}{n-1} \cdot \frac{p^n}{(1-p)}$.

For $p > 1/2$, we show that $\E{\ALG(\cI)} / \E{\OPT(\cI)} \leq \frac{1}{2}\left(1+\frac{5}{n-1}\right)$. Otherwise,
\begin{align*}
    \E{\ALG(\cI)} > \E{\OPT(\cI)}\frac{1}{2}\left(1+\frac{5}{n-1}\right) > \frac{p}{2}\left(1+\frac{5}{n-1}\right)> \frac{1}{4}\left(1+\frac{5}{n-1}\right)\OPT,
\end{align*}
where the second inequality follows from the fact that $\E{\OPT(\cI)} \geq p \OPT$.
Now since $\ALG(\cI)$ is equivalent to an order-oblivious algorithm with a sampling phase of length $Binom(n,p)$, this contradicts Theorem~\ref{thm:secretary_ob_upperbound}. 
\end{proofw}

\section{Replacing Greedy with the Maximum Matching}\label{apx:opt_based}
Inspired by the optimal algorithm for bipartite matching with vertex arrivals in the random-order model~\cite{DBLP:conf/esa/KesselheimRTV13}, one could hope to improve upon the order-oblivious competitive-ratio of the greedy-based algorithm by using the offline maximum matching instead of the offline greedy matching at each online round. We show that this approach leads to a non-competitive algorithm. We call the algorithm obtained from Algorithm~\ref{alg:greedy_v_a} by replacing $\Greedy$ with $\OPT$, the \textit{opt-based} algorithm. 

\begin{theorem}\label{thm:non_competitive}
For any $p \in [0,1]$, there is an infinite sequence of bipartite graphs $\CS{G}_1, \CS{G}_2, \dots$ such that the competitive-ratio of the opt-based algorithm on $\CS{G}_{k}$ approaches $0$ as $k \rightarrow{\infty}$. 
\end{theorem}

\begin{proof}
Let $\CS{G}_k =\CS{(L,R,E)}$ where $\CS{R}=\set{r_1,r_2,\dots, r_k}$ and $\CS{L}=\set{v, u_1,\dots,u_k}$. $v$ is connected to all vertices in $\CS{R}$ with $w(v,r_i) = a$ for all $1 \leq i\leq k$ (one should think of $a$ as being arbitrarily large). Each $u_i$ is connected only to $r_i$ through an edge of weight $w(u_i,r_i) = 1$.
 
We choose $a$ large enough so that $k/a$ is negligible. 
Clearly, $\OPT \geq a $. Let $\RS{O} = \CS{L} \setminus \RS{L}'$ be the set of vertices that arrive after the sampling phase. To upper bound the expected value of $\ALG$, note that when $v \notin \RS{O}$, the maximum matching in $\CS{G}[\RS{O}]$ is upper bounded by $k$. Therefore, in case $v \notin \RS{O}$, $\ALG$ obtains a negligible value compared to $\OPT$. We can therefore focus on the case $v \in \RS{O}$. 

Observe that if $\set{u_1,\dots,u_k} \cap \RS{O} \neq \emptyset$, then $v$ is matched in $\OPT(\CS{G}[\RS{L'}\cup \set{v}])$ to some $r_i$ such that $u_i \notin \RS{\RS{L'}}$. This is because the maximum matching includes as many vertices as possible from $\set{u_1,\dots,u_k} \cap \RS{L'}$ in addition to $v$. Consider the order in which the vertices in $\set{u_1,\dots,u_k} \cap \RS{O}$ arrive first, and $v$ arrives afterwards. Each $u_i \in \RS{O}$ is matched by the algorithm to $r_i$. Then, when $v$ arrives, by the observation above the selected candidate edge is $(v,r_i)$ for some $r_i$ such that $u_i \notin \RS{L'}$ (in other words, $u_i \in \RS{O}$). Therefore, $r_i$ is already occupied by $u_i$ which arrived earlier. We have $\E{\given{\ALG}{\set{u_1,\dots,u_k} \cap \RS{O} \neq \emptyset}  } \leq k $. For the second case, $\set{u_1,\dots,u_k} \cap \RS{O} = \emptyset$, we use the trivial upper-bound of $(a+k)$ on the expected value of $\ALG$. For convenience of notation, let $A$ denote the event $\set{u_1,\dots,u_k} \cap \RS{O} = \emptyset$. We get that
\begin{align*}
    \E{\ALG} &= \E{\given{\ALG}{A}}\Pr\left[A\right] + \E{\given{\ALG}{\neg{A}}}\Pr\left[\neg{A}\right] \leq (a+k)\cdot \Pr[ A] +  k\cdot 1 
    \leq a p^k + 2k.
\end{align*}
Hence, $\E{\ALG}/\OPT \leq \frac{ap^k + 2k}{ a}$. Since $2k/ a$ is arbitrarily small, we get that $\E{\ALG}/\OPT= o(1)$ (note that if $p=1$, then $\E{\ALG} = 0$). 
\end{proof}

\end{document}